\documentclass[11pt]{article}
\usepackage{amsmath}
\usepackage{amssymb}
\usepackage{amsthm}

\usepackage{xcolor}

\usepackage{geometry}

\usepackage{vmargin}
\usepackage{authblk}

\newtheorem*{pro*}{Proposition}
\newtheorem{lem}{Lemma}
\newtheorem{pro}{Proposition}
\newtheorem{rmk}{Remark}
\newtheorem{col}{Corollary}

\newtheorem{thm}{Theorem}
\newtheorem*{thm*}{Theorem A}
\newcommand{\bm}{\boldsymbol}
\newcommand{\Tr}{\mathop{\mathrm{Tr}}}
\newcommand{\C}{\mathbb{C}}
\newcommand{\D}{\mathcal{D}}
\newcommand{\res}{\mathop{\mathrm{res}}} 
\usepackage{color}
\title{On algebro-geometric solutions to the Gelfand--Dickey hierarchy}
\date{}
\author{Zejun Zhou \thanks{zzj24601@mail.ustc.edu.cn}}
\affil{ School of Mathematical Sciences, University of Science and Technology of China, Hefei	230026, China}

\begin{document}
	\maketitle
	\begin{abstract}
		In \cite{Du19} Dubrovin introduced an $A_1$-type infinite ODE system and gave a simple way of constructing algebro-geometric solutions to the KdV hierarchy (cf. also	\cite{Du20AlgCur,BDY}). 
		In \cite{InfiODE} the infinite ODE system is generalized to $\mathfrak{g}$-type infinite ODE system, where $\mathfrak{g}$ is any simple Lie algebra. In this paper, we give a simple constructinon of algebro-geometric solutions to the Gelfand--Dickey hierarchy based on the $A_n$-type infinite ODE system and Dubrovin's method.
		As an application, we give a formula for the $N$-point function for the related Riemann $\theta$-function.
	\end{abstract}
	\section{Introduction}\label{sec1}
	Let $n\geq 1$ be an integer. The Gelfand–Dickey hierarchy with $n$ unknown functions is an infinite family of PDEs \cite{Dickey}, defined by
	\begin{equation}\label{GelfandDickey}
		\frac{\partial L}{\partial T^a_k}=[(L^{\frac{a}{n+1}+k})_+,L]\,,\quad a=1,\dots,n\,, k\geq 0\,.
	\end{equation}
	Here
	\begin{equation}
		L:=\partial^{n+1}+v_1\partial^{n-1}+v_2\partial^{n-2}+\dots+v_n\,,
	\end{equation}
	and $\partial$ is understood as $\partial_{T^1_0}$ (see \cite{Dickey} for details). When $n=1$, it is the celebrated Korteweg--de Vries (KdV) hierarchy. 
	
	In \cite{Dubrovin1974,Novikov1974a},  Dubrovin and Novikov (see also Matveev and Its \cite{Its1975}) developed a way to give a family of exact solutions to the KdV hierarchy; see also \cite{Mumford83}. These are called algebro-geometric solutions to the KdV hierarchy, also known as finite-gap solutions. In their method, the KdV hierarchy is transformed into evolution of a point in the Jacobian of a hyperelliptic curve, which is the spectral curve of the differntial operator $L$.  In \cite{Dickson1999,Dickson1999b}, Dickson--Gesztesy--Unterkofler gave algebro-geometric solutions to the Boussinesq hierarchy, and the relating spectral curves are trigonal curves. This is the case when $n=2$.  There are also studies on algebro-geometric solutions related to trigonal curves and  tetragonal curves (e.g. \cite{Zeng2025,Xu2024}).  When $n\geq3$, as far as we know, the algebro-geometric solutions to the Gelfand--Dickey hierarchy seem not to be given through explicit spectral curves.
	
	 In \cite{Krichever1976,Krichever1977} for every compact Riemann surface, Krichever constructed algebro-geometric solutions to the Kadomtsev--Petviashvili (KP) hierarchy using $\theta$-function; see also \cite{Enolski2011, KNTY88_CFT, Nakayashiki2016,SegalWilson1985}. The Gelfand--Dickey hierarchy is a reduction of the KP hierarchy, so in principle one can use Krichever's method to construct the algebro-geometric solutions to the Gelfand--Dickey hierarchy. However, explicit constructions still require study.

	 Recently Dubrovin \cite{Du19} introduced an infinite commuting  ODE system with infinitely many unknowns in terms of a matrix Laurent series $W(\lambda)\in \mathfrak{sl}_{2}(\C)((\lambda^{-1}))$. 
	    If $W(\lambda)$ is truncated at the $g$-th term then the corresponding solution to the ODE system gives \cite{Du19} an algebro-geometric solution to the KdV hierarchy, which is related to the spectral curve of $\lambda^gW(\lambda)$：
	\begin{equation}
		C:\mu^2=\det(\lambda^gW(\lambda))=\lambda^{2g+1}+\sum_{k=1}^{2g+1}q_k\lambda^{2g+1-k}\,.
	\end{equation}
	The corresponding $\tau$-function can be expressed by $\theta$-functions of this spectral curve (cf.~\cite{Du19,Nakayashiki01}). 
Some interesting properties of $\theta$-functions are obtained in \cite{Du19} based on also the matrix-resolvent method (cf.~\cite{BDY,BDY2}). 
 This new way from \cite{Du19} to get algebro-geometric solutions (cf. also \cite{Du20AlgCur}) is direct and simpler.

	In \cite{DS}  the KdV hierarchy is generalized to Drinfeld--Sokolov hierarchy of $\mathfrak{g}$ type, where $\mathfrak{g}$ is any simple Lie algebra and the KdV hierarchy corresponds to $\mathfrak{sl}_2(\C)$. 
		In \cite{InfiODE} Dubrovin's infinte ODE system is generalized from $\mathfrak{sl}_2(\C)$ to $\mathfrak{g}$. It was proved in \cite{InfiODE} that solutions to the infinite ODE system lead to $\tau$-functions \cite{BDY2} of the Drinfeld--Sokolov hierarchy of $\mathfrak{g}$ type.
	
 As suggested in \cite{InfiODE}, it is possible to extend Dubrovin's new way \cite{Du19,Du20AlgCur} of constructing algebro-geometric solutions to the Drinfeld--Sokolov hierarchy of $\mathfrak{g}$ type. In this paper we are going to achieve this in the $\mathfrak{g}=A_{n}$ case.
  
  Let $W(\lambda)\in \mathfrak{sl}_{n+1}(\C)[\lambda,\lambda^{-1}]$ be given by
\begin{equation}\label{WlambdaLaurentPoly}
	W(\lambda)=\Lambda(\lambda)+\sum_{k\geq 0}^{m}\frac{W_k}{\lambda^k}\,.
\end{equation}
Here 
\begin{equation}
	\Lambda(\lambda)=\lambda E_{n+1,1}+\sum_{k=1}^nE_{k,k+1}
\end{equation}
is the cyclic element, and $W_0$ is a lower trangular matrix. The spectral curve of $\lambda^mW(\lambda)$ is given by
\begin{equation}\label{SpectralCurve}
C:	S(\lambda,\mu):=\det(\mu I_{n+1}-\lambda^mW(\lambda))=0\,.
\end{equation}
In the generic situation, its genus is $g=\frac{mn(n+1)}{2}$. Let $a_1,\dots,a_g,b_1,\dots,b_g\in H_1(C,\mathbb{Z})$ be its canonical cycles, and $\omega_1,\dots,\omega_g$ be the corresponding normalized holomorphic differentials on $C$.
 In Section \ref{secTau} we will show that $C$ has only one point $\infty_C$ at infinity, and the local coordinate can be taken as \begin{equation}\label{def_z}
 	z:=\lambda^{-1/(n+1)}\,.
 \end{equation}
 Define the vectors $\bm V^{(k)}=(V^{(k)}_1,\dots,V^{(k)}_g)$, $k\geq 1$ by 
\begin{equation}\label{expand_omega}
\omega_i(z)=-\sum_{k\geq 1}V^{(k)}_iz^{k-1}dz\,,\quad i=1,\dots, g\,,k\geq 1\,.
\end{equation} 
Let $\omega(P,Q)$, $P,Q\in C$ be the fundamental normalized 
bi-differential (see \cite{Fay}). Let $z,w$ be the local coordinates of $P,Q$ when $P,Q$ is near $\infty_C$. Define $q_{i,j}$, $i,j\geq 1$ by 
\begin{equation}\label{expand_bidiffer}
	\omega(z,w)=\frac{dzdw}{(z-w)^2}+\sum_{i,j\geq 1}q_{i,j}z^{i-1}w^{j-1}dzdw\,.
\end{equation}

There is a natrual line bundle $\mathcal{L}$ of the eigenvector of $W(\lambda)$ on $C$, and in Section \ref{secTau} we will prove that the degree of $\mathcal{L}$ is $-g-n$. We can identify degree zero divisors with their image under the Abel--Jacobi map. Let $-\mathcal{D}_{\mathcal{L}}$ be a divisor of $\mathcal{L}$, we define
\begin{equation}\label{Defu}
	\bm u:=\mathcal{D}_{\mathcal{L}}-(n+1)\infty_C-\Delta\in J(C)\,.
\end{equation}
Here $\Delta$ is the Riemann divisor and its degree is $g-1$.

Following Dubrovin \cite{Du20AlgCur} we will prove  in Section \ref{secTau}
 \begin{pro}\label{ThmGDtau}
	 The function $Z(t_1,t_2,\dots)$ defined by
	\begin{equation}\label{GDtau}
		Z(t_1,t_2,\dots):=\exp\left(\sum_{i,j\geq 1}\frac{1}{2}q_{i,j}t_it_j\right)\theta\left(\sum_{k\geq 1}t_k\bm V^{(k)}-\bm u\right)\,
	\end{equation}
	is the $\tau$-function of the solution to the infinite ODE in \cite{InfiODE} corresponding to $W(\lambda)$ (and so is a $\tau$-function of the Gelfand--Dickey hierarchy). Here $\theta$ is the $\theta$-function \cite{Dubrovin1981,Du19,Du20AlgCur}  associated to \eqref{SpectralCurve}, and $t_{a+(n+1)k}:=T^a_k$, $a=1,\dots,n$, $k\geq 0$. 
\end{pro}

	In \cite{Du19,Du20AlgCur} Dubrovin proved a formula for the logrithmic derivatives of $\theta$-functions coming from algebro-geometric solutions to the KdV hierarchy and the $N$-wave hierarchy.  Moreover, he used these formulas to show that certain combinations of the Taylor coefficients of $\log\theta$ at certain points are all rational numbers.
	Let us also generalize these to $A_n$.

	  Following \cite{Du20AlgCur}, define a matrix-valued function
		\begin{equation}\label{DefPhiP}
		\Phi(P):=\frac{S(\lambda,W)-S(\lambda,\mu)}{W-\mu}\,,\quad P=(\lambda,\mu)\in C\,.
	\end{equation}
Here we take the change $W(\lambda)\mapsto \lambda^mW(\lambda)$.	Define $\Pi(P)$ by 
	\begin{equation}\label{Projector}
		\Pi(P):=\frac{\Phi(P)}{S_{\mu}(\lambda,\mu)}\,,\quad P=(\lambda,\mu)\in C\,.
	\end{equation}
The following theorem will be proved in Section \ref{secProof1}.
	\begin{thm}\label{mainCol}
	 For any $n\geq 1$, $N\geq 2$, let $\lambda_1,\dots,\lambda_N$ be distinct complex numbers and let $P_1=(\lambda_1,\mu_1),\dots,P_N=(\lambda_N,\mu_N)$ be $N$ points on $C$ given in \eqref{SpectralCurve}. Define $\omega_N$ by
	\begin{equation}\label{omega_m}
		\omega_N(P_1,\dots,P_N):=\sum_{i_1,\dots,i_N=1}^{g}\frac{\partial^N \log\theta(\bm u)}{\partial u_{i_1}\dots\partial u_{i_N}}\omega_{i_1}(P_1)\dots\omega_{i_N}(P_N)+\delta_{N,2}\omega(P_1,P_2)\,,
	\end{equation}
	where $\bm u$ is defined by \eqref{Defu}.
	Then we have 
	\begin{equation}\label{multidifferential}
		\omega_N(P_1,\dots,P_N)=-\frac{1}{N}\sum_{s\in S_N}\frac{\Tr\bigl(\Pi(P_{s_1})\dots\Pi(P_{s_N})\bigr)d\lambda_1\dots d\lambda_N}{\prod_{i=1}^{N}(\lambda_{s_{i+1}}-\lambda_{s_{i}})}\,.
	\end{equation}
Here $S_N$ is the permutation group of $N$ elements, and $s_{N+1}:=s_{1}$.
\end{thm}
For the case when $n=1$, Theorem \ref{mainCol} was proved in \cite{Du19}.

 As a direct corollary of Theorem \ref{mainCol}, we have the following 
	\begin{col}\label{Rational}
	Let $W(\lambda)\in\mathfrak{sl}_{n+1}(\C)[\lambda,\lambda^{-1}]$ be given by \eqref{WlambdaLaurentPoly}. If every entry of $W(\lambda)$ belongs to $\mathbb{Q}[\lambda,\lambda^{-1}]$ and $N\geq 3$, then all the $N$-th Taylor coefficients of $\log \theta\left(\sum_{k\geq 1 }t_k\bm V^{(k)}-\bm u\right)$ at $\bm t=\bm 0$ are rational numbers.
\end{col}
For the case when $n=1$, Corollary \ref{Rational} was proved in \cite{Du19}.

In \cite{Du19} Dubrovin also proved that every $\tau$-function of the KdV hierarchy can be approximated by $\theta$-functions of hyperelliptic curves. We will generalize this result to the Gelfand--Dickey hierarchy.

	The paper is organized as follows. In Section \ref{secTau} we will prove Proposition \ref{ThmGDtau} by solving the ODE defined in \cite{InfiODE}. In Section \ref{secProof1} we wil prove Theorem \ref{mainCol}.  Section \ref{sec7} is on an example.

	\section{The algebro-geometric $\tau$-functions for the Gelfand--Dickey hierarchy}\label{secTau}
	In this section we will prove Proposition \ref{ThmGDtau}. Our proof is mainly along the lines in \cite{Du19,Du20AlgCur} and is based on \cite{InfiODE}.
	
	\subsection{The ODE system of $W(\lambda)$ }
	
	Let us review 	the ODE system introduced in \cite{InfiODE} 
	for an element $W(\lambda)\in\mathfrak{sl}_{n+1}(\mathbb{C})((\lambda^{-1}))$:
	\begin{equation}\label{WlambdaLuarent}
		W(\lambda)=\Lambda(\lambda)+W_0+\sum_{k=1}^{\infty}\frac{W_k}{\lambda^k}\,.
	\end{equation} 
	We always assume $W_0$ is lower triangular. 
	Define \cite{InfiODE} a Poisson bracket on $\mathcal{P}:= \mathbb{C}[(W_{k})_{ij}|k\geq 1\, \text{or} \,k=0, i\geq j]$ by 
	\begin{equation}\label{PoiBra}
		\left\{W(\lambda) \underset{'}{\otimes} W(\mu)\right\}:=[P(\lambda-\mu),W(\lambda)\otimes 1+1\otimes W(\mu)].
	\end{equation}
	Here 
	\begin{equation}
		P(\lambda)=\sum_{i,j=1}^{n+1}\frac{E_{ij}\otimes E_{ji}}{\lambda}\,.
	\end{equation}
	The principal gradation of $\mathfrak{sl}_{n+1}(\mathbb{C})((\lambda^{-1}))$ is 
	\begin{equation}\label{PrincipalGrad}
		\deg E_{ij}=j-i\,,\quad \deg \lambda=n+1\,.
	\end{equation} 
	It has been proved \cite{InfiODE} that for any $W(\lambda)$ there exist unique pair $(U(\lambda), H(\lambda))$ of the form:
	\begin{align}
		U(\lambda)=\sum_{k\geq1}U^{[-k]}(\lambda)\in {\rm Im}\,{\rm ad_{\Lambda(\lambda)}}\,,\quad \deg U^{[-k]}(\lambda)=-k\,,\\
		H(\lambda)=\sum_{k\geq 0}H^{[-k]}(\lambda)\in {\rm Ker} \,{\rm ad_{\Lambda(\lambda)}}\,, \quad \deg H^{[-k]}(\lambda)=-k\,,
	\end{align}
	satisfying 
	\begin{equation}\label{DecomW}
		e^{-{\rm ad}_{U(\lambda)}}W(\lambda)=\Lambda(\lambda)+H(\lambda)\,.
	\end{equation}
	The matrix $H(\lambda)$ can be uniquely written as 
	\begin{equation}\label{HamiltonianDef}
		H(\lambda)=\sum_{a=1}^{n}\sum_{k\geq-1}\frac{h_{a,k}}{n+1}\Lambda(\lambda)^{n+1-a}\lambda^{-k-2}\,.
	\end{equation}
	Here $h_{a,k}\in\mathcal{P}$.   The {\it basic matrix resolvents of $W(\lambda)$} are defined as
	\begin{equation}\label{DefMatResolvent}
		R_a(\lambda):=e^{{\rm ad}_{U(\lambda)}}\Lambda(\lambda)^a\,,\quad a=1,\dots,n\,.
	\end{equation}
	Define the hamiltonian flows 
	\begin{equation}\label{ODE}
		\frac{d W(\lambda)}{dT^a_k}=\left\{h_{a,k},W(\lambda)\right\}\,.
	\end{equation}
	All the flows $\frac{d}{d T^a_k}$, $a=1,\dots,n$, $k\geq 0$ commute with each other, so the ODE system \eqref{ODE} has a solution $W(\lambda,\bm T)\in \mathfrak{sl}_{n+1}(\C)[\lambda,\lambda^{-1}][[\bm T]]$, for any initial value. Here $(\bm T)=(T^a_k)_{a=1,\dots,n,\, k\geq 0}$.
	The $\tau$-structures $F_{a,k;b,l}$, $a,b=1,\dots,n$, $k,l\geq 0$ of the ODE system \eqref{ODE} are defined by 
	\begin{equation}\label{DefTaustr}
		\sum_{k,l\geq 0}\frac{F_{a,k;b,l}}{\lambda^{k+1}\nu^{l+1}}=\frac{\Tr(R_a(\lambda)R_b(\nu))}{(\lambda-\nu)^2}-\delta_{a+b,n+1}\frac{a\lambda+b\nu}{(\lambda-\nu)^2}\,,\quad a,b =1,\dots,n\,.
	\end{equation}
	For any solution to the ODE system \eqref{ODE}, there exist a $\tau(\bm T)$,  such that
	\begin{equation}\label{DefTau}
		F_{a,k;b,l}(\bm T)=\frac{\partial ^2\log \tau(\bm T)}{\partial T^a_k\partial {T^b_l}}\,.
	\end{equation}
	This $\tau(\bm T)$ is called the $\tau$-function of the solution $W(\lambda,\bm T)$, and is proved to be a $\tau$-function of the Gelfand--Dickey hierarchy \cite{InfiODE}. 
	
	Let $N\geq 2$, define polynomials $F_{a_1,k_1;\dots,a_N,k_N}\in \mathcal{P}$ by
	\begin{equation}
		F_{a_1,k_1;\dots;a_N,k_N}:=\frac{d^N\log \tau}{dT^{a_1}_{k_1}\dots dT^{a_N}_{k_N}}.
	\end{equation}
	It was proved \cite{InfiODE} that
	\begin{align}
		\sum_{k_1,\dots,k_N\geq 0}\frac{F_{a_1,k_1;\dots;a_N,k_N}}{\lambda_1^{k_1+1}\cdots\lambda_N^{k_N+1}}=-\frac{1}{N}\sum_{s\in S_N}\frac{\Tr\bigl(R_{a_{s_1}}(\lambda_{s_1})\dots R_{a_{s_N}}(\lambda_{s_N})\bigr)}{\prod_{i=1}^{N}(\lambda_{s_i}-\lambda_{s_{i+1}})}\nonumber\\
		-\delta_{N,2}\delta_{a_1+a_2,n+1}\frac{a_1\lambda_1+a_2\lambda_2}{(\lambda_1-\lambda_2)^2}\,.\label{mPointCorrelator}
	\end{align}
	Therefore for a solution $W(\lambda,\bm T)$, we have an explicit expression for the generating series of the $N$-th order Taylor coefficients of $\tau(\bm T)$:
	\begin{align}
		\sum_{k_1,\dots,k_N\geq 0}\frac{\frac{\partial^N\log \tau}{\partial T^{a_1}_{k_1}\dots \partial T^{a_N}_{k_N}}(\bm 0)}{\lambda_1^{k_1+1}\cdots\lambda_N^{k_N+1}}=&-\frac{1}{N}\sum_{s\in S_N}\frac{\Tr\bigl(R_{a_{s_1}}(\lambda_{s_1},\bm 0)\dots R_{a_{s_N}}(\lambda_{s_N},\bm 0)\bigr)}{\prod_{i=1}^{N}(\lambda_{s_i}-\lambda_{s_{i+1}})}\nonumber\\
		&-\delta_{N,2}\delta_{a_1+a_2,n+1}\frac{a_1\lambda_1+a_2\lambda_2}{(\lambda_1-\lambda_2)^2}\,,\quad N\geq 2\,.\label{mtauPointCorrelator}
	\end{align}

	For $n=1$  the ODE restricted to truncated $W(\lambda)$ was considered \cite{Du19}. This restricted ODE system also appeared  in  \cite{Nakayashiki01}.
	Inspired by these work, we consider the case when $W(\lambda)$ is given by \eqref{WlambdaLaurentPoly}. 
	It has been shown in \cite{FT-book} that for any positive $M,N$, the subspace 
	\begin{equation}\label{TruancatedPoissonSub}
		\left\{W(\lambda)\in \mathfrak{sl}_{n+1}(\C)((\lambda^{-1}))|W(\lambda)=\sum_{k=-M}^{N}W_k\lambda^k\right\}
	\end{equation} is a Poisson submanifold of $\mathfrak{sl}_{n+1}(\C)((\lambda^{-1}))$. Therefore the space 
	\begin{equation}
		\left\{W(\lambda)\in \mathfrak{sl}_{n+1}(\C)((\lambda^{-1}))|	W(\lambda)=\Lambda(\lambda)+W_0+\sum_{k=1}^{m}\frac{W_k}{\lambda^k}\right\}
	\end{equation} 
	is an intersection of Poisson submanifolds  \eqref{TruancatedPoissonSub} and the space of $W(\lambda)$ given in \eqref{WlambdaLuarent}, so itself is also a Poisson submanifold.
	Therefore the ODE system \eqref{ODE} can be restricted to this subspace, and becomes a finite dimensional hamiltonian system. The solutions and $\tau$-functions of this ODE system will be discussed in the subsequent subsections.

	\subsection{Spectral curve and the line bundle of eigenvector}
In this subsection and the next, for a $W(\lambda)$ given by \eqref{WlambdaLaurentPoly}, we take the change:
\begin{equation}\label{changeW}
	W(\lambda)\mapsto \lambda^mW(\lambda)\,.
\end{equation}	
Recall that the corresponding spectral curve $C$ is defined in \eqref{SpectralCurve}, which is of the form
	\begin{equation}\label{algebraicCurveDef}
		C:\,S(\lambda,\mu)=\mu^{n+1}-\lambda^{m(n+1)+1}+\sum_{k=1}^{n-1}\sum_{j=0}^{m(n+1-k)}s_{k,j}\lambda^j\mu^{k}=0\,.
	\end{equation}
	This curve is the so-called $(n+1,m(n+1)+1)$-curve; see \cite{BL}.
 Here and below we suppose it has no singularities on its affine part.
 
 For a fixed $\lambda$, let $\mu_1(\lambda), \dots, \mu_{n+1}(\lambda)$ be $n+1$ roots of equation \eqref{algebraicCurveDef}. As $\lambda\rightarrow\infty$, they satisfy following  formula:
	\begin{equation}\label{muAsym}
		\mu_k(\lambda)=e^{2k\pi\sqrt{-1}/(n+1)}\lambda^{m+\frac{1}{n+1}}+o\left(\lambda^{m+\frac{1}{n+1}}\right)\,, \quad k=1,\dots, n+1\,.
	\end{equation}
	So the cover $C\xrightarrow{\lambda}\mathbb{P}^1$ is ramified at a point $\infty_C$ above the infinity in $\mathbb{P}^1$, and the ramification index of $\lambda$ at $\infty_C$ is $n+1$. Around $\infty_C$ we have this parameterization:

	\begin{equation}\label{InfiLocalCor}
		\left\{\begin{aligned}
			&\frac{1}{\lambda}=z^{n+1}\,,\\
			&\frac{\mu}{\lambda^{m}}=\frac{1}{z}+O(1)\,,\quad z\rightarrow0\,.
		\end{aligned}\right.
 	\end{equation}
	Using Riemann--Hurwitz formula we can prove the genus of $C$ is  $g=\frac{mn(n+1)}{2}$.

	For a point $P=(\lambda,\mu)\in C$, we can associate it with an eigenvctor~$\boldsymbol{\psi}(P)\in\mathbb{C}^{n+1}$of $W(\lambda)$:
	\begin{equation}\label{eignenvec}
		W(\lambda)\boldsymbol{\psi}(P)=\mu \boldsymbol{\psi}(P)\,.
	\end{equation}
	If $P$ is not a ramification point, then $\boldsymbol{\psi}(P)$ lies in an one-dimensional subspace. When $P$ is a ramification point, let $\zeta$ be a local coordinate around $P$ and $\zeta(P)=0$. Then the meromorphic eigenvector $\bm \psi(\zeta)$ can be written as 
	\begin{equation}\label{psi_r}
		\bm \psi(\zeta)=\zeta^d\tilde{\bm \psi}(\zeta)\,, \quad \tilde{\bm\psi}(0)\in \C^{n+1}\backslash\bm 0\,.
	\end{equation}
Let $\tilde{\bm \psi}(0)$ be the eigenvector associated to $P\in C$. This complete the definition of the line budle $\mathcal{L}$ of $W(\lambda)$'s eigenvector on $C$. The integer $d$ in \eqref{psi_r} is the order of the meromorphic section $\bm \psi(\zeta)$ at $P$.

	Let $P=(\lambda_0,\mu_0)\in C\backslash \infty_C$ be a ramification point. Denote by $s$ ramification index of $\lambda$ at $P$. Thus we can choose local coordinate $\zeta=(\lambda-\lambda_0)^{\frac{1}{s}}$. Let $\boldsymbol{\psi}(\zeta)$ be a local section around $P$:
	\begin{equation}\label{RamiEigenVec}
		W(\zeta^s+\lambda_0)\boldsymbol{\psi}(\zeta)=\mu(\zeta)\boldsymbol{\psi}(\zeta)\,, \quad {\rm ord}_P\boldsymbol{\psi}(\zeta)=0\,.
	\end{equation} 
Since $P$ is not a singular point, we have the following expansion as $\zeta\rightarrow0$:
	\begin{equation}
		\mu(\zeta)=\mu_0+\mu'(0)\zeta+o(\zeta)\,,\quad \mu'(0)\neq0\,.
	\end{equation}
	Now differentiate~\eqref{RamiEigenVec} $k$ times at the point $\zeta=0$, ($k<s$). There is 
	\begin{equation}\label{W-muRecurRelation}
		W(\lambda_0)\boldsymbol{\psi}^{(k)}(0)=\sum_{l=0}^{k}\binom{k}{l}\mu^{(k-l)}(0)\boldsymbol{\psi}^{(l)}(0)\,.
	\end{equation}
	Here $*^{(k)}$ means taking $k$-th derivative with respect to $\zeta$. 
	\begin{lem}\label{lemRootspace}
		The root space of $W(\lambda_0)$ with eigenvalue $\mu_0$ has basis:
		\begin{equation}
			\boldsymbol{\psi}^{(0)}(0),\dots,\boldsymbol{\psi}^{(s-1)}(0)\,.
		\end{equation}
		Here the derivative $*^{(k)}$ is with respect to local coordinate $\zeta=(\lambda-\lambda_0)^{\frac{1}{s}}$ around $P$.
	\end{lem}
	\begin{proof}
		The root space of $W(\lambda_0)$'s eigenvalue $\mu_0$ is $s$-dimensional, and we see from \eqref{W-muRecurRelation} that  $\boldsymbol{\psi}^{(0)}(0),\dots,\boldsymbol{\psi}^{(s-1)}(0)$ belong to this space. So it is sufficient to prove that they are linearly independent.
		Suppose we have a linear relation 
		\begin{equation}
			\sum_{k=0}^{s-1}c_k\bm \psi^{(k)}(0)=\bm 0\,.
		\end{equation}
		Let $M$ be the biggest number such that $c_M\neq 0$. Because \eqref{W-muRecurRelation}, we have 
		\begin{equation}
			(W(\lambda_0)-\mu_0)^{M}\left(\sum_{k=0}^{s-1}c_k\bm \psi^{(k)}(0)\right)=c_MM!(\mu'(0))^M\bm \psi(0)=\bm 0\,.
		\end{equation}
		This result shows $c_M=0$, because $\mu'(0)\neq0$ and $\bm \psi(0)\neq \bm 0$. Therefore we get a contradiction, which means $c_k=0$, $k=0,\dots,s-1$. So $\boldsymbol{\psi}^{(0)}(0),\dots,\boldsymbol{\psi}^{(s-1)}(0)$ are linearly independent. Hence the lemma is proved.
	\end{proof}

Now consider a local section $\bm{\psi}(z)=(\psi_1(z),\dots,\psi_{n+1}(z))^T$ around $\infty_C$. Recall that $z=\lambda^{-1/(n+1)}$.  One can verify that 
\begin{equation}\label{expandPsiInfinity}
	\psi_k(z)=z^{d-k+1}+o\left(z^{d-k+1}\right)\,,\quad k=1,\dots,n+1\,.
\end{equation}
Here $d$ equals to ${\rm ord}_{\infty_C}\bm \psi(z)+n$.

Recall that $\mathcal{L}$ is the line bundle of the eigenvector of $W(\lambda)$. We can use the method in \cite[Chapter 5]{Babelon} to show  that the degree of $\mathcal{L}$ is $-g-n$.

	Now fix an eigenvector $\bm{\psi}(P)$ by condition:
	\begin{equation}\label{nomalizedPsi}
		\bm{\psi}(P)=(1,\psi_2(P),\dots,\psi_{n+1}(P))^{T}\,,\quad P\in C\,.
	\end{equation}
	Here and below we call this $\bm{\psi}(P)$ the {\it normalized eigenvector of W($\lambda$).}
	Its expansion at $\infty_C$ is
	\begin{equation}\label{normPsiExpandsion}
		\psi_i(z)=z^{-i+1}+o\left(z^{-i+1}\right)\,,\quad z\rightarrow 0\,, \quad i=1,\dots n+1\,.
	\end{equation}
	So this eigenvector $\bm \psi(P)$ has an $n$-th order pole at $\infty_C$, and rest of $g$ poles $Q_1,\dots,Q_g$ on the affine part of $C$.
Define the divisor $\mathcal{D}$ by
\begin{equation}\label{DefDivisorD}
	\mathcal{D}:=Q_1+\dots+Q_g\,.
\end{equation}
	
It can be proved that $l(\mathcal{D})=\dim H^0(C,\mathcal{O}(\mathcal{D}))=1$, which means that $\mathcal{D}$ is a non-special divisor.

		Let $G$ be a lower triangular matrix, and the diagonal entries are all $1$. We call $W(\lambda)\mapsto GW(\lambda)G^{-1}$  a gauge transformation of $W(\lambda)$. It is clear that $C$ and $\mathcal{D}$ does not change under gauge transformations.

	Up to now, for a $W(\lambda)$ as the form \eqref{WlambdaLaurentPoly}, we have associated it to its spectral curve $C$ and a degree $g$ non-special divisor $\mathcal{D}=Q_1+\dots+Q_g$. Conversely, given the data $C$ and $\mathcal{D}$, we can determine $W(\lambda)$ up to a gauge transformation.

	\begin{lem}\label{reConLem}
		Let $C$ be an algebraic curve given by \eqref{algebraicCurveDef}. Let  $\mathcal{D}=Q_1+\dots+Q_g$ be a non-special divisor on C, and $Q_1,\dots,Q_g \neq \infty_C$. Let  $G=(g_{ij})$ be a $(n+1)\times (n+1)$ lower triangular matrix, such that $g_{ii}=1$, $i=1,\dots,n+1$. Then there exist a unique matrix  $W(\lambda)$ of the form \eqref{WlambdaLaurentPoly}, whose spectral curve is $C$. And the normalized eigenvector $\bm \psi(P)$
		has poles at $Q_1,\dots,Q_g$, while the Laurent expansion of $\bm{\psi}(P)$ at $\infty_C$ is 
		\begin{equation}\label{psiExpansion2}
			\psi_i(z)=\sum_{k=1}^{n+1}\frac{g_{ik}}{z^{k-1}}+O\left(z\right)\,.
		\end{equation}
	\end{lem}
	\begin{proof}
		 Now because $\mathcal{D}$ is a non-special divisor, using the Riemann--Roch Theorem we conclude that 
		\begin{equation}
			l(\mathcal{D}+k\infty_C)=k+1\,,\quad k=0,\dots,n+1\,.
		\end{equation}
		Therefore we can choose $\psi_1,\dots,\psi_{n+1}\in L(\mathcal{D}+n\infty_C)$ such that satisfy \eqref{psiExpansion2}. Here $L(\mathcal{D}+n\infty_C):=\left\{f|(f)+\mathcal{D}+n\infty_C\geq 0\right\}$. Obviously $\psi_1(P)=1$. Denote $\bm{\psi}=(\psi_1,\dots,\psi_{n+1})^{T}$.
		For any $\lambda\in \mathbb{P}^1$, let $P_1=(\lambda, \mu_1),\dots,P_{n+1}=(\lambda,\mu_{n+1})$ be $n+1$ points on $C$ over $\lambda$. Define the matrix:
		\begin{equation}
			\Psi(\lambda):=(\bm{\psi}(P_1),\dots,\bm{\psi}(P_{n+1}))\,.
		\end{equation}
		Let $\hat{\mu}(\lambda)={\rm diag}(\mu_1,\dots,\mu_{n+1})$. The matrix-value function $\Psi(\lambda)\hat{\mu}(\lambda)\Psi(\lambda)^{-1}$ is independent of the order of $P_1,\dots,P_{n+1}$, so is a well-defined meromorphic function of $\lambda$ on $\mathbb{P}^1$. Take the Laurent expansion of $\Psi(\lambda)\hat{\mu}(\lambda)\Psi(\lambda)^{-1}$ at $\infty$:
		\begin{equation}
			\Psi(\lambda)\hat{\mu}(\lambda)\Psi(\lambda)^{-1}=\sum_{k=0}^{N}W_k\lambda^k+O(\lambda^{-1})=W(\lambda)+O(\lambda^{-1})\,.
		\end{equation}
		Here $W(\lambda)$ is the polynomial part in the expansion. Now let us consider a vector-value meromorphic function on $C$ 
		\begin{equation}
			\tilde{\bm{\psi}}(P)=W(\lambda)\bm{\psi}(P)-\mu(P)\bm{\psi}(P)\,,\quad P=(\lambda,\mu)\in C\,.
		\end{equation}
		The vector function $\tilde{\bm{\psi}}(P)$ could only has poles at $Q_1,\dots,Q_g$, and its limit at $\infty_C$ equals to $\bm{0}$. Because $\mathcal{D}$ is a non-special divisor, we conclude that $\tilde{\bm{\psi}}=\bm{0}$. So $C$ is the spectral curve of $W(\lambda)$ and $\bm \psi (P)$ is the eigenvector of $W(\lambda)$. From expansion \eqref{muAsym} and \eqref{psiExpansion2} we conclude that $W(\lambda)$ is of the form \eqref{WlambdaLaurentPoly}. The uniqueness of $W(\lambda)$ can be deduced from the formula 
	$	W(\lambda)=\Psi(\lambda)\hat{\mu}(\lambda)\Psi(\lambda)^{-1}$.
	\end{proof}
	One can repeat the above construction to the dual line bundle $\mathcal{L}^\dagger$ of the dual eigenvector of $W(\lambda)$:
	\begin{equation}
		\bm{\psi}^\dagger(P) W(\lambda)=\mu\bm{\psi}^\dagger(P)\,.
	\end{equation}
	Consider the normalized dual eigenvector.
	\begin{equation}\label{dualEigenvec}
		\bm{\psi}^\dagger(P)=(\psi^\dagger_1,\dots,\psi^\dagger_{n},1)\,.
	\end{equation}
	Similar to the above argument, we have
	\begin{lem}\label{lemDivisorDagger}
		The normalized dual eigenvector $\bm \psi^\dagger$ has expansion at $\infty_C$:
		\begin{equation}\label{dualPsiExpansion}
			\psi^{\dagger}_j(z)=z^{-n-1+j}+o\left(z^{-n-1+j}\right)\,,\quad j=1,\dots,n+1\,.
		\end{equation}
		And $\bm \psi^\dagger$ has $g$ poles  $Q_1^\dagger,\dots,Q_{g}^\dagger$ on the affine part of $C$. The divisor \begin{equation}\label{defDivsorDagger}
			\mathcal{D}^\dagger=Q_1^\dagger+\cdots+Q_{g}^\dagger
		\end{equation} is also a non-special divisor.
	\end{lem}

		\subsection{$\theta$-function formulas}
 	We will derive formulas that relate $W(\lambda)$ to $\theta$-functions in this subsection.
 	First we prove the following generalization of \cite[Lemma 2.6]{Du20AlgCur}.
 	\begin{pro}\label{proSpecProjector}
 		Suppose  $\lambda\in\mathbb{P}^1$ is not a branch point, and $P_i(\lambda,\mu_i)$, $i=1,\dots,n=1$ are $n+1$ distinct points on $C$ over $\lambda$. Denote $\Pi_i(\lambda):=\Pi(\lambda,\mu_i)$, then we have
 		\begin{equation}
 			\Pi_i\cdot\Pi_j=\delta_{i,j}\Pi_i\,,\quad \sum_{i=1}^{n+1}\Pi_i(\lambda)=I_{n+1}\,,\quad\sum_{i=1}^{n+1}\mu_i\Pi_i(\lambda)=W(\lambda)\,.
 		\end{equation}
 	\end{pro}
 	\begin{proof}
 		By using the same method as that in  \cite{Du20AlgCur}.
 	\end{proof}
 	
 	The above proposition means that $\Pi_i$, $i=1,\dots,n+1$ are projectors that send column vectors in $\C^{n+1}$ to $\ker(\mu_iI-W(\lambda))$. Now we can express the basic matrix resolvents of $W(\lambda)$ by $\Pi_i$, $i=1,\dots,n+1$. When $\lambda\rightarrow\infty$, the $n+1$ points $P_1,\dots,P_{n+1}$ have coordinates
 	\begin{equation}\label{cordinatePj}
 		z(P_j)=e^{\frac{2\pi j\sqrt{-1}}{n+1}}\lambda^{\frac{1}{n+1}}\,,\quad j=1,\dots,n+1\,.
 	\end{equation}
 	By equation \eqref{DecomW}, \eqref{HamiltonianDef} and \eqref{DefMatResolvent}, we know that $R_a(\lambda)$ have the same eigenvectors as $W(\lambda)$ does. Let $\bm \psi(z)$ be the eigenvector of $W(\lambda)$ defined on a small neighbourhood of $\infty_C$. Then the following equation holds
 	\begin{equation}
 		R_a(\lambda)\bm \psi(z)=z^{-a}\bm \psi(z)\,,\quad a=1,\dots n\,.
 	\end{equation}
 	So $R_a(\lambda)$ can be expressed by $\Pi_{j}$:
 	\begin{equation}\label{RaPij}
 		R_a(\lambda)=\sum_{j=1}^{n+1}z(P_j)^{-a}\Pi_{j}\,.
 	\end{equation}
 	 When $P\rightarrow\infty_C$,
 	by equation \eqref{cordinatePj} and \eqref{RaPij}, the expansion of $\Pi(P)$ is as follows:
 	\begin{equation}\label{OmegaAndR}
 		\Pi(z)=\frac{1}{n+1}\left(\sum_{a=1}^{n}z^aR_a(\lambda)+I_{n+1}\right)\,.
 	\end{equation}

 	Let $\Omega(P):=\Pi(P)d\lambda$. By equation \eqref{mPointCorrelator} we have the following expansion for $N\geq 2$.
 		\begin{align}\label{OmegazFakm}
 			&	-\frac{1}{N}\sum_{s\in S_N}\frac{\Tr\bigl(\Omega(z_{s_1})\dots\Omega(z_{s_N})\bigr)}{\prod_{i=1}^{N}(\lambda_{s_{i+1}}-\lambda_{s_{i}})}-\delta_{N,2}\frac{dz_1dz_2}{(z_1-z_2)^2}=\nonumber\\
 			&\sum_{a_1,\dots,a_N=1}^{n}\sum_{l_1,\dots,l_N=0}^{\infty}F_{a_1,l_1;\dots; a_N,l_N}z_1^{a_1+(n+1)l_1-1}\dots z_N^{a_N+(n+1)l_N-1}dz_1\dots dz_N\,.
 		\end{align}

 	Now we need to express entries of $\Omega(P)$ by the spectral curve $C$ and the divisor $\mathcal{D}$. 

Because $\Phi(P)$ is a polynomial of $W(\lambda)$ and $\mu$, it must be holomorphic on $C\backslash\infty_C$.  So the entries of $\Omega(P)$ are holomorphic on $C\backslash\infty_C$. 	Let $\bm{\psi}$, $\bm\psi^\dagger$ be the normalized eigenvector and dual eigenvector. Proposition \ref{proSpecProjector} tells us 
\begin{equation}\label{Omega_ij}
	\Omega_{ij}=\frac{\psi_i\psi^\dagger_j }{\bm{\psi}^\dagger\bm{\psi}}d\lambda\,.
\end{equation}
 Now using formula \eqref{normPsiExpandsion} and \eqref{dualPsiExpansion} we get  
 		\begin{equation}\label{OmegaExpand}
 			\Omega_{ij}(z)=-(1+o(z))\frac{dz}{z^{n+2+i-j}}\,.
 		\end{equation}

 	From the above equation we see that $\Omega_{1,n+1}$ has a double pole at $\infty_C$. From the definition $\Omega(P):=\Pi(P)d\lambda$ and Proposition \ref{proSpecProjector} we know that $\psi_j(P)=\frac{\Omega_{j,n+1}(P)}{\Omega_{1,n+1}(P)}$.
 	Because $\Omega_{ij}$ has no poles on the affine part of $C$, we conclude that $Q_1,\dots,Q_g$ must be zeros of $\Omega_{1,n+1}$. Similarily $Q_1^\dagger,\dots,Q_g^\dagger$ are also zeros of $\Omega_{1,n+1}$. So the $2g$ zeros of $\Omega_{1,n+1}$ are exactly $Q_1,\dots,Q_g, Q^\dagger_1,\dots,Q^\dagger_{g}$. This kind of differentials can be expressed by $\theta$-functions, and following are the details.

 	Recall that we have introduced
 	\begin{equation}
 		\bm{u}:=\mathcal{D}_{\mathcal{L}}-(n+1)\infty_C-\Delta=\mathcal{D}-\infty_C-\Delta\in J(C)\,
 	\end{equation} 
 	as a point in Jacobian corresponding to the line bundle $\mathcal{L}$.  Because the canonical divisor on $C$ is 
 	\begin{equation}
 		K_C=2\Delta=(\Omega^1_{n+1})=\D+\D^\dagger-2\infty_C\,.
 	\end{equation}
 	So we also have 
 	\begin{equation}
 		-\bm{u}=\mathcal{D}^\dagger-\infty-\Delta\,.
 	\end{equation}
 	
 An important theorem of Riemann (cf. \cite{Fay}) states:
 		Let $\bm u\in J(C)$, and let $\Delta$ be the Riemann divisor on $C$, then
 		\begin{enumerate}
 			\item If $\theta(\bm u)\neq 0$, then the function $\theta(P-\infty_C-\bm u)$ of $P\in C$ has $g$ zeros $P_1,\dots, P_g$. The divisor $\mathcal{A}=P_1+\dots+P_g$ is non-special and satisfies
 			\begin{equation}
 				\bm u=\mathcal{A}-\infty_C-\Delta\,.
 			\end{equation}
 			\item If $\theta(\bm u)=0$, then there exist an effective divisor $\mathcal{B}$ of degree $g-1$ such that
 			\begin{equation}
 				\bm u=\mathcal{B}-\Delta\,.
 			\end{equation}
 		\end{enumerate}
 	
 	Let $E(P,Q)$ be the prime-form, see \cite{Fay}.
 	\begin{equation}
 		E(P,Q)=\frac{\theta[\nu](Q-P)}{\sqrt{\sum\omega_i(P)\partial_{u_i}\theta[\nu](0)}\sqrt{\sum\omega_k(Q)\partial_{u_k}\theta[\nu](0)}}\,.
 	\end{equation}
 	Here $\nu$ is a nonsingular odd half-period. Let $\zeta$ be a local coordinate in a neighbourhood of a point $p\in C$, and $P,Q$ be two points in this neighbourhood. Then near $Q=P$, the prime form could be expressed as 
 	\begin{equation}\label{primeformExpansion}
 		E(P,Q)=\frac{\zeta(Q)-\zeta(P)}{\sqrt{d\zeta(P)d\zeta(Q)}}(1+S(p)(\zeta(Q)-\zeta(P))^2+\text{higher order terms})\,,
 	\end{equation}
 	where $S(p)$ is a constant determined by $p$.

 The following lemma is a generalization of \cite[Proposition 2.14]{Du20AlgCur}.
 	\begin{lem}\label{lemOmega1n+1}
 		The differential form $\Omega_{1,n+1}$ could be written as 
 		\begin{equation}\label{Omega1n+1}
 			\Omega_{1,n+1}(P)=-\frac{\theta(P-\infty_C+\bm{u})\theta(P-\infty_C-\bm{u})}{\theta(\bm{u})^2E(P,\infty_C)^2dz(\infty_C)}\,.
 		\end{equation}
 	\end{lem}
 	\begin{proof}
 		From Riemann's theorem, we know that $\theta(\bm u)\neq 0$ because $\mathcal{D}$ is non-special. The $g$ zeros of $\theta(P-\infty_C-\bm u)$ are just $Q_1,\dots,Q_g$. Similarly, zeros of $\theta(P-\infty_C+\bm u)$ are $Q_1^\dagger,\dots,Q_g^\dagger$.
 		So we can verify that the right-hand side is a well-defined differential form having the same zeros and poles as $\Omega_{1,n+1}(P)$. Then by comparing the coefficients of expansion at $\infty_C$ we conclude the identity.
 	\end{proof}
 	 Meromorphic functions in $L(\mathcal{D}+n\infty_C)$ are given in the following lemma.
 	\begin{lem}\label{meromorphicWaveFuc}
 		Let 
 		\begin{equation}
 			\phi_i(P)=\frac{\left.\partial_{z(Q)}^{i-1}\frac{\theta(P-Q-\bm{u})}{\theta(\bm{u})E(P,Q)\sqrt{d z(Q)}}\right|_{Q=\infty_C}}{(i-1)!\frac{\theta(P-\infty_C-\bm{u})}{\theta(\bm{u})E(P,\infty_C)\sqrt{d z(\infty_C)}}}\,, \quad i=1,\dots,n+1\,.
 		\end{equation}
 		 Then $\phi_i(P)$ are well-defined meromorphic functions on $C$, and are basis for $L(\mathcal{D}+n\infty_C)$. Moreover when $P$ is near $\infty_C$, we have
 		\begin{equation}\label{phiAsy}
 			\phi_i(z)=\frac{1}{z^{i-1}}+o\left(\frac{1}{z^{i-1}}\right)\,.
 		\end{equation}
 	\end{lem}
 	\begin{proof}
 		First we can verify that $\phi_i(P)$, $i=1,\dots,n+1$ are well-defined meromorphic functions on $C$. Because $\theta(P-\infty_C-\bm u)$  only has zeros at $Q_1,\dots, Q_g$, we know $(\phi_i(P))+ \mathcal{D}+N\infty_C\geq 0$. Here $N$ is an large positive integer. Now we consider $\phi_i(P)$ when $P $  is near $\infty_C$.
 		By the property of the prime-form \eqref{primeformExpansion}, we have
 		\begin{equation}\label{derivitive_phi}
 			\phi_i(z)=\frac{\sqrt{dz}}{(i-1)!\frac{\theta(P-\infty_C-\bm{u})}{\theta(\bm{u})E(P,\infty_C)\sqrt{d z(\infty_C)}}}\left.\partial_w^{i-1}\frac{\theta(P-Q-\bm u)}{\theta(\bm u)h(z,w)(w-z)}\right|_{Q=\infty_C}\,.
 		\end{equation}
 		Here $w:=z(Q)$, and 
 		\begin{equation}
 			E(z,w)=\frac{w-z}{\sqrt{dzdw}}h(z,w)\,,\quad h(z,w)=1+\text{higher order terms}\,.
 		\end{equation} 
 		So 
 		\begin{equation}
 			\frac{\theta(P-Q-\bm u)}{\theta(\bm u)h(z,w)}
 		\end{equation}
 		equals to $1$ when $z=w=0$, and is holomorphic when $z,w$ are small.
 		Because
 		\begin{equation}
 			\left.\partial_w^{i-1}\frac{\theta(P-Q-\bm u)}{\theta(\bm u)h(z,w)(w-z)}\right|_{w=0}=-\sum_{k=0}^{i-1}\binom{i-1}{k}\frac{(i-1-k)!}{z^{i-k}}\left.\partial_w^{k}\frac{\theta(P-Q-\bm u)}{\theta(\bm u)h(z,w)}\right|_{w=0}\,,
 		\end{equation}
 		we have following  from equation \eqref{derivitive_phi}
 		\begin{equation}
 			\phi_i(z)=\sum_{k=0}^{i-1}\frac{\left.-\partial_w^{k}\frac{\theta(P-Q-\bm u)}{\theta(\bm u)h(z,w)}\right|_{w=0}}{-k!\frac{\theta(P-\infty_C-\bm u)}{\theta(\bm u)h(z,0)z^{-1}}}\cdot \frac{1}{z^{i-k}}=\frac{1}{z^{i-1}}+o\left(\frac{1}{z^{i-1}}\right)\,.
 		\end{equation}
 		So equation \eqref{phiAsy} is proved.
 		Now it is easy to see that $\phi_i(P)\in L(\mathcal{D}+n\infty_C)$, $i=1,\dots,n+1$, and they form a basis of $L(\mathcal{D}+n\infty_C)$. 
 	\end{proof}
 	Similarly 
 	\begin{equation}
 		\phi_j^\dagger(P)=\frac{\left.\partial_{z(Q)}^{n+1-j}\frac{\theta(P-Q+\bm{u})}{\theta(\bm{u})E(P,Q)\sqrt{d z(Q)}}\right|_{Q=\infty_C}}{(n+1-j)!\frac{\theta(P-\infty_C+\bm{u})}{\theta(\bm{u})E(P,\infty_C)\sqrt{d z(\infty_C)}}}\,, \quad j=1,\dots,n+1\,,
 	\end{equation} 
 	are basis of the space $L(\mathcal{D}^\dagger+n\infty_C)$.
 	Now let us consider the inner product:
 	\begin{equation}		
 		\sum_{i=1}^{n+1}\phi_i(Q)\phi^\dagger_i(P)=\frac{\left.\partial^n_{z(R)}\frac{\theta(P-R+\bm{u})\theta(Q-R-\bm{u})}{\theta(\bm{u})^2E(P,R)E(Q,R)dz(R)}\right|_{R=\infty_C}}{n!\frac{\theta(P-\infty_C+\bm{u})\theta(Q-\infty_C-\bm{u})}{\theta(\bm{u})^2E(P,\infty_C)E(Q,\infty_C)dz(\infty_C)}}\,.
 	\end{equation}

 		Difine the differential 
 		\begin{equation}
 			H_{PQ}(R):=\lambda(R)\frac{\theta(P-R+\bm{u})\theta(Q-R-\bm{u})}{\theta(\bm{u})^2E(P,R)E(Q,R)}\,.
 		\end{equation}
 		It has poles at the point $R=P,Q,\infty_C$. The sum of its residues  is zero, so 
 		\begin{align}
 			{\res}_{R=\infty_C}H_{PQ}(R)&=-({\res}_{R=P}H_{PQ}(R)+{\res}_{R=Q}H_{PQ}(R))\nonumber\\
 			&=-\lambda(P)\cdot \frac{\theta(Q-P-\bm{u})}{\theta(\bm u)E(Q,P)}-\lambda(Q)\frac{\theta(P-Q+\bm{u})}{\theta(\bm u)E(P,Q)}\nonumber\\
 			&=(\lambda(P)-\lambda(Q))\frac{\theta(P-Q+\bm u)}{\theta(\bm u)E(P,Q)}\,.
 		\end{align}
 		Then we have
 		\begin{align}
 			&\frac{\theta(P-\infty_C+\bm{u})\theta(Q-\infty_C-\bm{u})}{\theta(\bm{u})^2E(P,\infty_C)E(Q,\infty_C)dz(\infty_C)}\sum_{i=1}^{n+1}\phi_i(Q)\phi^\dagger_i(P)\nonumber\\
 			=&\left.\frac{1}{n!}\partial^n_{z(R)}\frac{\theta(P-R+\bm{u})\theta(Q-R-\bm{u})}{\theta(\bm{u})^2E(P,R)E(Q,R)dz(Q)}\right|_{R=\infty_C}\nonumber\\
 			=&{\res}_{R=\infty_C}H_{PQ}(R)\nonumber\\
 			=&(\lambda(P)-\lambda(Q))\frac{\theta(P-Q+\bm u)}{\theta(\bm u)E(P,Q)}\,.\nonumber
 		\end{align}
 		Therefore we have the following interesting formula.
 		 \begin{equation}\label{innerProduct}
 			\frac{\theta(P-\infty_C+\bm{u})\theta(Q-\infty_C-\bm{u})}{\theta(\bm{u})^2E(P,\infty_C)E(Q,\infty_C)dz(\infty_C)}\sum_{i=1}^{n+1}\phi_i(Q)\phi^\dagger_i(P)=(\lambda(P)-\lambda(Q))\frac{\theta(P-Q+\bm u)}{\theta(\bm u)E(P,Q)}\,.
 		\end{equation}
 	Suppose $\lambda$ is not a branch point,  and let $P_1,\dots,P_{n+1}\in C$ be the $n+1$ distinct points over $\lambda$. Then we have
 	\begin{equation}\label{orthgnal}
 		\sum_{i=1}^{n+1}\phi_i(P_j)\phi^\dagger_i(P_k)=0\,,\quad j\neq k\,.
 	\end{equation}

 	Because $Q_1,\dots,Q_g$ are poles of the normalized eigenvector $\bm \psi(P)$, we have
 	\begin{equation}\label{psiaik}
 		\psi_i=\sum_{k=1}^{n+1}a_{ik}\phi_k\,\quad i=1,\dots,n+1\,.
 	\end{equation}
 	By \eqref{normPsiExpandsion} and \eqref{phiAsy} we know that $(a_{ik})_{1\leq i,k\leq n+1}$ is a lower triangular matrix and its diagnal elements are $1$.
 	Let $P=(\lambda,\mu)\in C$, the vector 
 	\begin{equation}
 		\bm{\phi}(P)=(\phi_1,\dots,\phi_{n+1})^T
 	\end{equation}
 	is an eigenvector of the matrix $A^{-1}W(\lambda)A$ with the eigenvalue $\mu$, where $A=(a_{ik})$. Identity \eqref{orthgnal} shows that 
 	\begin{equation}
 		\bm{\phi}^\dagger(P)=(\phi_1^\dagger,\dots,\phi_{n+1}^\dagger)
 	\end{equation}
 	is the dual eigenvector of $A^{-1}W(\lambda)A$. So
 	\begin{equation}\label{psidaggerAphi}
 		\bm{\psi}^\dagger A=\bm{\phi}^\dagger \,.
 	\end{equation}
 	And equation \eqref{psiaik} could be written as 
 	\begin{equation}\label{psiAphi}
 		\bm{\psi}=A\bm{\phi}\,.
 	\end{equation}
 	
 	The following proposition is the special case $N=2$ of Theorem \ref{mainCol}. This proposition generalizes \cite[Theorem 1.1]{Du20AlgCur}
	\begin{pro}\label{OmegaPro}
 		Let $P$, $Q$ be two distinct points on $C$, then 
 		\begin{equation}
 			\frac{{\Tr}(\Omega(P)\Omega(Q))}{(\lambda(P)-\lambda(Q))^2}=\frac{\theta(P-Q+\bm{u})\theta(P-Q-\bm{u})}{\theta(\bm{u})^2E(P,Q)^2}\,.
 		\end{equation}
 	\end{pro}
 	\begin{proof}
 		From equation \eqref{psidaggerAphi} and \eqref{psiAphi}, we have 
 		\begin{equation}\label{innerProductPsiPhi}
 			\bm{\psi}^\dagger(P)\bm{\psi}(Q)=\bm{\phi}^\dagger(P)\bm{\phi}(Q)\,.
 		\end{equation}
 	  The proposition can be proved (cf. e.g. \cite{Dubrovin2020}) by Lemma \ref{lemOmega1n+1} and equation \eqref{innerProduct}.
 	\end{proof}

 	Fay \cite{Fay} proved the identity:
 	\begin{equation}
 		\frac{\theta(P-Q+\bm{u})\theta(P-Q-\bm{u})}{\theta(\bm{u})^2E(P,Q)^2}=\sum_{j,k=1}^g\frac{\partial^2\log\theta(\bm u)}{\partial u_j\partial u_k}\omega_j(P)\omega_k(Q)+\omega(P,Q)\,.
 	\end{equation}
 	Combining Proposition \ref{OmegaPro}, we have
 	\begin{equation}\label{OmegaPQDDlogTheta}
 		\frac{{\Tr}(\Omega(P)\Omega(Q))}{(\lambda(P)-\lambda(Q))^2}=\sum_{j,k=1}^g\frac{\partial^2\log\theta(\bm u)}{\partial u_j\partial u_k}\omega_j(P)\omega_k(Q)+\omega(P,Q)\,.
 	\end{equation}

 	\subsection{Proof of Proposition \ref{ThmGDtau}}
 	In this subsection we are going to describe the $\tau$-functions of the solutions to  the ODE system \eqref{ODE} when restricted to $W(\lambda)$ given by \eqref{WlambdaLaurentPoly}.
 	
 	Recall that $t_{a+(n+1)k}:=T^a_k$. We use $W(\lambda,\bm t)$, $\bm t=(t_1,t_2,\dots)$ to denote a solution to \eqref{ODE}. Here we suppose $W(\lambda, \bm t)$ does not depend on $t_{(n+1)k}$, $k\geq 1$. 
 
 		In following discussion we use the notation:
 		\begin{equation}\label{DefMk}
 			M_{k}(\lambda):=(R_1(\lambda)^k)_+ \,,\quad k\geq1\, .
 		\end{equation}
 		Here and further we use $(*)_+$ to denote the polynomial part. Recall that $R_a(\lambda)=R_1(\lambda)^a$, $a=1,\dots,n$, and $R_1(\lambda)^{n+1}=\lambda I_{n+1}$.
 		By result in \cite{InfiODE} and the condition that $W(\lambda,\bm t)$ is independent of $t_{(n+1)k}$, $k\geq 1$, we get the following equations from the ODE \eqref{ODE}.
 		\begin{equation}\label{ODEham}
 			\frac{d W(\lambda,\bm t)}{d t_k}=[M_k(\lambda),W(\lambda,\bm t)]\,,\quad k\geq 1\,.
 		\end{equation}
 		Therefore the spectral curve $C$ of $W(\lambda,\bm t)$ stay unchanged under the flow $\frac{d}{dt_k}$.

 		Let $\bm \psi(P,\bm t)$ be the normalized eigenvector of $W(\lambda,\bm t)$. Consider the ODE system of $\bm v(P,\bm t)=(v^1(P,\bm t),\dots,v^{n+1}(P,\bm t))$:
 		\begin{equation}\label{ODEof_v}
 			\frac{d \bm v(P,\bm t)}{d t_k}=M_k(\lambda,\bm t)\bm v(P,\bm t)\,,\quad k\geq 1\,.
 		\end{equation}
 		In \cite{InfiODE} it has been proved that 
 		\begin{equation}
 			\frac{\partial M_k(\lambda,\bm t)}{\partial t_l}-\frac{\partial M_l(\lambda,\bm t)}{\partial t_k}+[M_k(\lambda,\bm t),M_{l}(\lambda, \bm t)]=0\,.
 		\end{equation}
 		So flows $\frac{d}{dt_k}$ in \eqref{ODEof_v} commute with each other. Therefore ODE system \eqref{ODEof_v} can be solved simutaneously. Suppose the initial value is 
 		\begin{equation}
 			\bm v(P,\bm 0)=\bm \psi(P,\bm 0)\,.
 		\end{equation}
 		Because 
 		\begin{equation}
 			\frac{d}{d t_k}\left(\left(W(\lambda,\bm t)-\mu\right)\bm v(P,\bm t)\right)=\bm 0\,.
 		\end{equation}
 		So $\bm v(P,\bm t)$ is an eigenvector of $W(\lambda,\bm t)$. We have
 		\begin{equation}\label{vPsi}
 			\bm v(P,\bm t)=v^1(P,\bm t)\bm \psi(P,\bm t)\,.
 		\end{equation}
 		Because $R_1(\lambda)$ has eigenvalue $z^{-1}$, for $k\geq 1$ we have 
 		\begin{equation}
 			(R_1(\lambda)^k)_+\bm \psi(P,\bm t)=z^{-k}\bm \psi(P,\bm t)-(R_1(\lambda)^k)_-\bm \psi(P,\bm t)\,.
 		\end{equation}
 		Therefore by \eqref{DefMk} we have
 		\begin{equation}
 			M_k(\lambda_,\bm t)\bm \psi(P,\bm t)=z^{-k}\bm \psi(P,\bm t)+o(1)\,,\quad P\rightarrow\infty_C\,, k\geq 1\,.
 		\end{equation}
 		
 		So when $P$ is near $\infty_C$ we have following expansion
 		\begin{equation}\label{vt_expand}
 			v^i(z,\bm t)= e^{\sum_{k\geq1}t_kz^{-k}}\left(\sum_{j=1}^{n+1}\frac{g_{ij}}{z^{j-1}}+O(z)\right)\,.
 		\end{equation}
 		In particular 
 		\begin{equation}
 			v^1(z,\bm t)=e^{\sum_{k\geq1}t_kz^{-k}}\left(1+O(z)\right).
 		\end{equation}
 		One can also verify $\bm v(P,\bm t)$ has only $g$ poles $Q_1(\bm 0), \dots, Q_g(\bm 0)$ on the affine part of $C$.
 			Let $\Omega^{(k)}$ be the normalized second kind differentials:
 		\begin{align}
 			&\Omega^{(k)}(z)=\frac{-k}{z^{k+1}}dz-\sum_{i\geq0}q_{i,k}z^{i-1}dz\,, \label{SecondDiff1}\\
 			&\oint_{a_i}\Omega^{(k)}=0\,, \quad \oint_{b_i}\Omega^{(k)}=V^{(k)}_i\,, i=1,\dots g\,.\label{SecondDiff2}
 		\end{align} 
 		The standard fact in the theory of Baker--Akhiezer shows that 
 		\begin{equation}
 			v^1(P,\bm t)=\exp\left(\sum_{k\geq 1}t_k\int_{\infty_C}^{P}\Omega^{(k)}\right)\frac{\theta\left(\int_{\infty_C}^{P}\bm \omega-\bm u+\sum_{k\geq 1}t_k\bm V^{(k)}\right)}{\theta\left(\int_{\infty_C}^{P}\bm \omega-\bm u\right)}\,.
 		\end{equation}
 Here the integral $\int_{\infty_C}^{P}\Omega^{(k)}$ is regarded as 
 	\begin{equation}
 		\int_{\infty_C}^{P}\Omega^{(k)}=\lim_{Q\rightarrow\infty_C}\int_{Q}^{P}\Omega^{(k)}+\frac{1}{z(Q)^k}\,.
 	\end{equation}
 		By Riemann's theorem, when $\bm u-\sum_{k\geq 0}t_k\bm V^{(k)}\in J(C)\backslash(\Theta)$, the function $v^1(P,\bm t)$ has $g$ zeros. By \eqref{vPsi}, the zeros are $Q_1(\bm t), \dots,Q_g(\bm t)$, and 
 		\begin{equation}\label{u_t}
 			\bm u(\bm t)=Q_1(\bm t)+\dots+Q_g(\bm t)-\infty_C-\Delta=\bm u-\sum_{k\geq 0}t_k\bm V^{(k)}\,.
 		\end{equation}
By \eqref{vt_expand} we get
 	\begin{align}	\psi_i(z,\bm t)
 		=&e^{\sum_{k\geq1}\big(t_kz^{-k}-t_k\int_{\infty_C}^{z}\Omega^{(k)}\big)}\frac{\theta(\int_{\infty_C}^{z}\bm \omega-\bm u)\left(\sum_{j=1}^{n+1}\frac{g_{ij}}{z^{j-1}}\right)}{\theta(\int_{\infty_C}^{z}\bm \omega-\bm u+\sum_{k\geq 1}t_k\bm V^{(k)})}+O(z)\label{psitExpand}\,,
 	\end{align}
 	where $g_{ij}$ are constants.

 		We note that the solution $W(\lambda,\bm t)$ to \eqref{ODE} exists whenever 
 		\begin{equation}
 			\bm u(\bm t)=\bm u-\sum_{k\geq 1}t_k\bm V^{(k)}\in J(C)\backslash(\Theta)\,.
 		\end{equation}
 		Here $(\Theta)$ is the notation of theta divisor in $J(C)$.

 	For a given spectral curve $C$, the evolution of $g$ poles can be described using their coordinates $(\lambda_i,\mu_i)$, $i=1,\dots,g$. The follwing equations can be viewed as analogue of the Dubrovin equations in \cite{Dubrovin1974}.
 	\begin{col}\label{proDuvrovinEq}
 		Let $\phi_j(\lambda,\mu)$, $j=1,\dots,g$ be $g$ different monomials of the form
 		\begin{equation}
 			\phi_j(\lambda,\mu)=\lambda^{ml-k}\mu^{n-l}\,, l=1,\dots,n\,,\quad k=1,\dots,ml\,.
 		\end{equation}
 		Then the differential forms $\alpha_j=\frac{\phi_j(\lambda,\mu)d\lambda}{S_\mu(\lambda,\mu)}$, $j=1,\dots,g$ are basis of the space of holomorphic differential forms on $C$. The evolution equations of poles $Q_j=(\lambda_j,\mu_j)$， $j=1,\dots,g$ of the normalized eigenvector are as follows.
 		\begin{equation}\label{DubrovinEquation}
 			\frac{d \bm \lambda^T}{d t_k}=A(\bm \lambda,\bm \mu)^{-1}\bm U^{(k)}\,, \quad k\geq1\,.
 		\end{equation}
 		Here $\bm \lambda=(\lambda_1,\dots,\lambda_g)$, $\bm \mu=(\mu_1,\dots,\mu_g)$. The entries of matrix $A(\bm \lambda,\bm \mu)$ are defined by
 		\begin{equation}
 			A(\bm \lambda,\bm \mu)_{i,j}:=\frac{\phi_i(\lambda_j,\mu_j)}{S_{\mu}(\lambda_j,\mu_j)}\,, \quad i,j=1,\dots,g\,.
 		\end{equation}
 		The vector $\bm U^{(k)}$ are given by
 		\begin{equation}\label{UkVec}
 			\bm U^{(k)}=(U^{(k)}_1,\dots,U^{(k)}_g)	^T\,,\quad U^{(k)}_j:={\res}_{\infty_C}z^{-k}\alpha_j\,.
 		\end{equation}
 	\end{col}
 	\begin{proof}
 		From \eqref{u_t} we have 
 		\begin{align}
 			\frac{	d \bm u}{dt_k}=\sum_{j=1}^{g}\bm \omega(Q_j)\frac{d Q_j}{dt_k}=-\bm V^{(k)}={\res}_{\infty_C}z^{-k}\bm\omega\,.
 		\end{align}
 		
 		Let $\bm \alpha=(\alpha_1,\dots,\alpha_g)^T$. We get
 		\begin{equation}
 			\sum_{j=1}^{g}\bm \alpha(Q_j)\frac{d Q_j}{dt_k}={\res}_{\infty_C}\left(z^{-k}\bm \alpha\right)=\bm U^{k}\,.
 		\end{equation}
 		So we have 
 		\begin{equation}
 			A(\bm \lambda,\bm \mu)\frac{d\bm \lambda^T}{d t_k}=\bm U^{(k)}\,.
 		\end{equation}
 		And the corollary is proved.
 	\end{proof}

		Now we can prove Proposition \ref{ThmGDtau}.

	\begin{proof}[Proof of Proposition \ref{ThmGDtau}]\label{ODEtau}
		Let $W(\lambda,\bm t)$ be the solution to \eqref{ODE}, whose initial value is $W(\lambda)$. 
		Combining \eqref{OmegaPQDDlogTheta} and \eqref{OmegazFakm} for $N=2$, we get 
		\begin{align}
			&\sum_{a,b=1}^{n}\sum_{k,l\geq0 }F_{a,k;b,l}(\bm t)z_1^{a+k(n+1)-1}z_2^{b+l(n+1)-1}dz_1dz_2\nonumber\\
			=&\sum_{i,j=1}^g\sum_{k,l\geq1}\frac{\partial^2\log\theta(\bm u(\bm t))}{\partial u_i\partial u_j}V^{(k)}_iV^{(l)}_jz_1^{k-1}z_2^{l-1}dz_1dz_2+\sum_{k,l\geq1}q_{k,l}z_1^{k-1}z_2^{l-1}dz_1dz_2 \,.\nonumber
		\end{align}
		Here we have used expansions \eqref{expand_omega} and \eqref{expand_bidiffer}. Comparing the coefficients of both sides of above equation, we have the following equations by \eqref{u_t}.
		\begin{align}
			F_{a,k;b,l}(\bm t)=&\sum_{i,j=1}^{g}V^{(a+k(n+1))}_iV^{(b+l(n+1))}_j\frac{\partial^2\log\theta(\sum_{k\geq 1}t_k\bm V^{(k)}-\bm u)}{\partial u_i\partial u_j}+q_{a+k(n+1),b+l(n+1)}\nonumber\\
			=&\frac{\partial^2 \log\theta(\sum_{k\geq 1}t_k\bm V^{(k)}-\bm u)}{\partial t_{a+(n+1)k}\partial t_{b+(n+1)l}}+q_{a+k(n+1),b+l(n+1)}\label{F_theta_t}\,, \quad a,b=1,\dots n\,,\: k,l\geq 0\,.
		\end{align}
		Because equations \eqref{SecondDiff1} and \eqref{SecondDiff2}, we have the fact that $\bm V^{((n+1)k)}=\bm 0, q_{k(n+1),l}=0$, $k,l\geq 0$. So the function $Z(\bm t)$ given by \eqref{GDtau} does not depend on $t_{(n+1)k}$, $k\geq 1$, and satisfies the definition \eqref{DefTau} of $\tau$-function.  
		The proposition is proved.
	\end{proof}

\begin{rmk}
	We remark that combining Krichever's method (cf.~\cite{Enolski2011,Nakayashiki01,SegalWilson1985}) and Lemma~\ref{reConLem} one can also obtain Proposition \ref{ThmGDtau}.
\end{rmk}

	\section{Applications }\label{secProof1}
In this section we will prove Theorem \ref{mainCol} and give some applications of our results.

				\begin{proof}[Proof of Theorem \ref{mainCol}]
						For the $\tau$-function given in Proposition \ref{ThmGDtau}, we have
					\begin{align}
						&\sum_{a_1,\dots,a_N=1}^{n}\sum_{l_1,\dots,l_N=0}^{\infty}F_{a_1,l_1;\dots; a_N,l_N}z_1^{a_1+(n+1)l_1-1}\dots z_N^{a_N+(n+1)l_N-1}dz_1\dots dz_N\nonumber\\	
						=&\sum_{j_1,\dots,j_N=1}^{g}\frac{\partial^N\log\theta(\bm u)}{\partial u_{j_1}\dots\partial u_{j_N}}\omega_{j_1}(z_1)\dots\omega_{j_N}(z_N)+\delta_{N,2}\left(\omega(z_1,z_2)-\frac{dz_1dz_2}{(z_1-z_2)^2}\right)\label{RHSofnPoinCorrelator}\,.\nonumber
					\end{align}
					Here we have used equation\eqref{expand_omega}, \eqref{expand_bidiffer} and the fact that
					\begin{equation}
						q_{k(n+1),l}=0\,,\quad \bm V^{(k(n+1))}=\bm 0\,,\quad k,l\geq 1\,.
					\end{equation}
					
					Combining  \eqref{OmegazFakm},  we have equation \eqref{multidifferential}.
					\end{proof}
					Now we give an application of Theorem \ref{mainCol}.

		Following \cite{Zhou2015}, for an arbitrary power series $\tau(\bm t)$, define 
	\begin{equation}
		G_{\tau,N}(\xi_1,\dots,\xi_N)d\xi_1\dots d\xi_N:=\sum_{k_1,\dots,k_N\geq 1}\left.\frac{\partial^N\log \tau(\bm t)}{\partial t_{k_1}\dots\partial t_{k_N}}\right|_{\bm t=\bm 0}\prod_{i=1}^{N}\frac{d\xi_i}{\xi_i^{k_i+1}}\,,\quad N\geq 1\,.
	\end{equation} 
According to \eqref{GDtau} and \eqref{omega_m}, the multi-differential $G_{Z,N}(\xi_1,\dots,\xi_N)d\xi_1\dots d\xi_N$ equals to
	\begin{equation}\label{Zt_NpointFunc}
		(-1)^N\omega_N(\xi_1,\dots,\xi_N)-\delta_{N,2}\frac{d\xi_1 d\xi_2}{(\xi_1-\xi_2)^2}\,,
	\end{equation}
where $\omega_N(\xi_1,\dots,\xi_N)$ is the local expansion of $\omega_N$ at $P_1=\dots=P_N=\infty_C$ with respect to $\xi_k:=z(P_k)^{-1}$, $k=1,\dots,N$.
	Let $\bm \psi^\dagger$, $\bm \psi$ be given by \eqref{nomalizedPsi} and \eqref{dualEigenvec}. Then by Theorem \ref{mainCol} we have
	\begin{equation}\label{omegaN_psi}
		\omega_N=-\frac{1}{N}\sum_{s\in S_N}\prod_{j=1}^{N}\frac{\bm \psi^\dagger(P_{s_{j+1}})\bm \psi(P_{s_j})}{\bm \psi^\dagger(P_{s_{j+1}})\bm \psi(P_{s_{j+1}})}\frac{d\lambda(P_{s_{j+1}})}{\lambda(P_{s_{j}})-\lambda(P_{s_{j+1}})}\,.
	\end{equation}
	By the similar arguments in the proof of Proposition \ref{OmegaPro}, we have the following formula:
\begin{equation}\label{BPQ}
	\frac{\bm \psi^\dagger(Q)\bm \psi(P)}{\bm \psi^\dagger(Q)\bm \psi(Q)}\frac{d\lambda(Q)}{\lambda(P)-\lambda(Q)}=\frac{\theta(Q-P+\bm u)}{\theta(\bm u)E(Q,P)}\frac{\theta(Q-\infty_C-\bm u)E(P,\infty_C)}{\theta(P-\infty_C-\bm u)E(Q,\infty_C)}\,.
\end{equation}
Now by \eqref{BPQ}, we have
\begin{equation}\label{multidi_theta}
		\omega_N=-\frac{1}{N}\sum_{s\in S_{N}}\prod_{j=1}^{N}B(P_{s_j},P_{s_{j+1}})\,, \quad N\geq 2\,,
\end{equation}
  where
\begin{equation}\label{DefBPP}
	B(P,Q):=\frac{\theta(Q-P+\bm u)}{\theta(\bm u)E(Q,P)}\,.
\end{equation}
This formula has been obtained for general Riemann surface in \cite{KNTY88_CFT}.
 Similar formula to \eqref{multidi_theta} in the case that the spectral curve is unramified at $\infty$ also appeared in \cite{Du20AlgCur}.
 It is easy to verify that $B(P,Q)$ has the following expansion at $P=Q=\infty_C$ with respect to $\xi:=z(P)^{-1}$, $\eta:=z(Q)^{-1}$:
 \begin{equation}\label{BkernelForm}
 	B(\xi,\eta)=\left\{
 	\begin{aligned}
 		&A(\xi,\eta)\sqrt{d\xi d\eta}\,, \quad &&i= j\,,\\
 		&\frac{\sqrt{d\xi d\eta}}{\xi-\eta}+A(\xi,\eta)\sqrt{d\xi d\eta}\,, \quad &&i\neq j\,,
 	\end{aligned}
 	\right.
 \end{equation}
 where
 \begin{equation}
 	A(\xi,\eta)=\sum_{i,j\geq0}A_{i,j}\xi^{-i-1}\eta^{-j-1}\,.
 \end{equation}
Then by Zhou's theorem \cite{Zhou2015,YangZhou25,Alexandrov2025}, equation \eqref{multidi_theta} implies Proposition \ref{ThmGDtau} in turn. Geometrically, $A_{i,j}$ are the affine coordinates of $Z(\bm t)$ (cf. \cite{Sato1983,Enolski2011,Zhou2013,BY2017}), also cf. \cite{Cafasso2018} for the analogues in simple Lie algebras.

					Following \cite{Du19}, for	
				\begin{equation}
				f(\bm t)=\sum_{n\geq 0}\frac{1}{n!}\sum_{i_1,\dots,i_n\geq 1}f_{i_1,\dots,i_n}t_{i_1}\dots t_{i_n}\,,
			\end{equation}
define the $d$-th truncation of this series by
	\begin{equation}
	[f(t)]_d=\sum_{n\geq 0}\frac{1}{n!}\sum_{i_1+\dots+i_n\leq d}f_{i_1,\dots,i_n}t_{i_1}\dots t_{i_n}\,.
\end{equation}
\begin{thm}\label{ApproximateThm}
Let $\tau(\bm t)\in \mathbb{C}[[\bm t]]$ be an arbitrary $\tau$-function of the Gelfand--Dickey hierarchy. Then  for any $d\geq1$, there exist an positive integer  $ m\leq [\frac{d+n-1}{n+1}]$, such that we can find an algebraic curve $C$ given by \eqref{algebraicCurveDef} and a $\bm u\in J(C)\backslash (\Theta)$, satisfying:

\begin{equation}
	[\log\tau(\bm t)]_{d}=[\log \theta(\sum_{k=1}^{d}t_k\bm V^{(k)}-\bm u)]_d+\alpha+\sum_{k=1}^{d}\beta_kt_k+\sum_{k_1+k_2\leq d}\gamma_{k_1,k_2}t_{k_1}t_{k_2}\,.
\end{equation}
Here $\alpha$, $\beta_k$ and $\gamma_{k_1,k_2}$ are suitable constants.

\end{thm}

					\begin{proof}
						Define the extended gradation on the Poisson algebra $\mathcal{P}$:
						\begin{equation}
							\deg_e (W_k)_{ij}=1+i-j+k(n+1)\,,
						\end{equation}
						where $k>0$ or $k=0$ and $i\geq j$.
						By defining the extended gradation, we make $W(\lambda)$ to be homogenous:
						\begin{equation}
							\deg_e W(\lambda)=1\,.
						\end{equation}
						So if $f,g\in\mathcal{P}$ are homogenous polynomials in $\mathcal{P}$, the Poisson bracket of them is also homogenous:
						\begin{equation}
							\deg_e\{f,g\}=\deg_e f+\deg_e g-n-2\,.
						\end{equation}
						Under extened gradation, we have
						\begin{equation}
							\mathcal{P}=\bigoplus_{j=1}^{\infty}\mathcal{P}^{[j]}\,.
						\end{equation}
					Here $\mathcal{P}^{[j]}$ denotes the homogeneous components of degree $j$ under the extended gradation.
						By definition \eqref{DefTaustr}, we have
						\begin{equation}
							\deg_e F_{a,k;b,l}=a+k(n+1)+b+l(n+1)\,.
						\end{equation}
						And by \eqref{mPointCorrelator} we also have
						\begin{equation}
							\deg_eF_{a_1,k_1;\dots,a_m,k_m}=\sum_{j=1}^{m}a_j+(n+1)k_j\,.
						\end{equation}
						Because for $(W_k)_{ij}\in \mathcal{P}$ we have
						\begin{equation}
							\deg_e (W_k)_{ij}\geq 1-n+k(n+1)\,.
						\end{equation}
					Let 
					\begin{equation}
						\mathcal{P}_m:=\C[(W_k)_{i,j}|1\leq k\leq m\,,{\text{ or}}\, k=0,i\geq j]\,.
					\end{equation}
						So we have 
						\begin{equation}\label{gradP}
							\mathcal{P}^{[d]}\subset \mathcal{P}_m\,,\:\text{for all}\: m\geq\left[\frac{d+n-1}{n+1}\right]\,.
						\end{equation}

						According to \cite{InfiODE}, any $\tau$-function of the Gelfand--Dickey hierarchy is a $\tau$-function of the solution $W(\lambda,\bm t)$ to the equation \eqref{ODE}. Let $W(\lambda)$ be the ininitial value,  so the corresponding $\tau$-structures are
						\begin{equation}
							F_{a_1,k_1;\dots;a_l,k_l}=\left.\frac{\partial^l\log \tau(\bm t)}{\partial t_{a_1+(n+1)k_1}\dots \partial_{a_l+(n+1)k_l}}\right|_{\bm t=\bm 0}\,.
						\end{equation}
						Denote the truncated $W(\lambda)$ by 
						\begin{equation}
							[W(\lambda)]_m:=\sum_{i=0}^{m}\frac{W_i}{\lambda^i}+\Lambda(\lambda)\,.
						\end{equation}
						Choose the smallest $m$ such that 
						\begin{equation}
							[W(\lambda)]_m=[W(\lambda)]_{[\frac{d+n-1}{n+1}]}\,.
						\end{equation}
						Because \eqref{gradP}, if
						\begin{equation}
							\deg_e F_{a_1,k_1,\;\dots;a_l,k_l}=\sum_{i=1}^{l}a_i+(n+1)k_i\leq d\,,\quad l\geq 2\,,
						\end{equation}
						then the value of $F_{a_1,k_1,\;\dots;a_l,k_l}$ does not change when we replace $W(\lambda)$ by $[W(\lambda)]_{m}$.
						This means 
						\begin{equation}
							\left.\frac{\partial^l\log \tau(\bm t)}{\partial t_{k_1}\dots t_{k_l}}\right|_{\bm t=\bm 0}\,,\quad\sum {k_i}\leq d\,,l\geq2\,,
						\end{equation} 
						does not change when we replace $\tau(\bm t)$ by the $\tau$-function of the solution $W_m(\lambda,\bm t)$ who has initial value 
						\begin{equation}
							W_m(\lambda,\bm 0)=[W(\lambda)]_{m}\,.
						\end{equation}
						In general cases the spectral curve of $\lambda^m[W(\lambda)]_m$ has no singularities on its affine part, and the $\tau$-function of $W_m(\lambda, \bm t)$ can be given by \eqref{GDtau}.
						So Theorem \ref{ApproximateThm} is proved.
					\end{proof}
					
					\section{Example}\label{sec7}
					
					Before displaying a concrete example, let us introduce an explicit algorithm developed in \cite{Diener1994} to compute the divisor $\mathcal{D}$ defined by \eqref{DefDivisorD}.  In this section we always suppose $W(\lambda)$ is of the truncated form \eqref{WlambdaLaurentPoly} and its spectral curve has no singularities on the affine part. 
					
					Let $e^*$ be the row vector
					\begin{equation}
						e*:=(1,0,0,\dots,0)\in \mathbb{C}^{n+1}\,.
					\end{equation}
					Define the polynomial
					\begin{equation}\label{Dlambda}
						D(\lambda):=\det\begin{pmatrix}
							e^* \\
							e^*W(\lambda)\lambda^m \\
							e^*W(\lambda)^2\lambda^{2m} \\
							\vdots \\
							e^*W(\lambda)^n\lambda^{nm}
						\end{pmatrix}\,.
					\end{equation}
					And define polynomials $r_{j,k}(\lambda)$, $j,k=1,\dots,n+1$ by
					\begin{equation}
						\Delta_{j,1}(\lambda,\mu)=\sum_{k=1}^{n+1}r_{j,k}(\lambda)\mu^{n+1-k}\,.
					\end{equation}
					Here $\Delta_{j,1}$ are cofactors of the matrix $\mu I_{n+1}-\lambda^mW(\lambda)$, therefore $r_{j,1}=\delta_{j,1}$.
					Following proposition shows how to compute the divisor $\mathcal{D}$ of the matrix $W(\lambda)$.
				\begin{pro*}[\cite{Diener1994}]\label{divisorAlg}
						Let $\mathcal{D}=Q_1+\dots+Q_g$, and $Q_i=(\lambda_i,\mu_i)$. Then $D(\lambda)$ is a polynomial of degree $g$, whose zeros are $\lambda_1,\dots,\lambda_g$.
						
						Suppose $\lambda_1,\dots,\lambda_g$ are distinct, then the rectangular matrix $C(\lambda)=(r_{j,k}(\lambda))_{1\leq j\leq n+1, 1\leq k\leq n}$, evaluated at $\lambda_1,\dots, \lambda_g$ are of rank $n$. For every $k=1,\dots,g$, let $C_k$ be a nonzero $n$-minor of the matrix $C(z_k)$.
						\begin{equation}
							C_k:=(C(\lambda_k)_{i_s,j})_{1\leq s,j\leq n}\,.
						\end{equation}
						And let $\hat{C}_k$ be the matrix obtained from $C_k$ by changing the last column
						\begin{equation}
							r_{i_s,n}\mapsto r_{i_s,n+1}\,\quad s=1,\dots,n\,.
						\end{equation}
						Then 
						\begin{equation}
							\mu_k=-\frac{\det \hat{C}_k}{\det C_k}\,.
						\end{equation}
					\end{pro*}
					
					 The $n=1$ case was  discussed in \cite{Du19}.

					We will consider the $n=2$ case in more details in this section.
					
					In this case let 
					\begin{equation}
						W(\lambda)=\begin{pmatrix}
							u_1(\lambda)+u_2(\lambda) & u_3(\lambda) & u_4(\lambda) \\
							u_5(\lambda) & -u_2(\lambda) & u_6(\lambda) \\
							u_7(\lambda) & u_8(\lambda )& -u_1(\lambda)
						\end{pmatrix}\,.
					\end{equation}
					Here 
					\begin{align}
						u_i(\lambda)=\lambda\delta_{i,7}+u_{i,-1}+\sum_{k=0}^{m-1}\frac{u_{i,k}}{\lambda^{k+1}}\,,\quad i=1,2,5,7,8\,,\\
						u_i(\lambda)=\delta_{i,3}+\delta_{i,6}+\sum_{k=0}^{m-1}\frac{u_{i,k}}{\lambda^{k+1}}\,, i=3,4,6\,.
					\end{align}
					The spectral curve of $W(\lambda)$ is a genus $g=3m$ Riemann surface. By Proposition \ref{divisorAlg} we calculate the polynomial \eqref{Dlambda}.
					\begin{align}
						D(\lambda)&=\lambda^{3m}\left(-u_1(\lambda)u_3(\lambda)u_{4}(\lambda)+u_2(\lambda)u_{3}(\lambda)u_4(\lambda)+u_3(\lambda)^2u_6(\lambda)-u_4(\lambda)^2u_8(\lambda)\right)\nonumber\\
						&=\lambda^{3m}+{\rm l.o.t.}\,.
					\end{align}
					Let $\lambda_{1},\dots, \lambda_{g}$ be $g$ zeros of $D(\lambda)$, and $\mu_k$, $k=1,\dots,g$ are given by the algorithm in Proposition \ref{divisorAlg}. Then $Q_k=(\lambda_k,\mu_k)\in C$ and $\mathcal{D}=Q_1+\dots+Q_g$ is the divisor corresponding to $W(\lambda)$.
					
					When $m=1$ the explicit formula is 
					\begin{align}
						D(\lambda)=&\lambda^{3}+\left(-u_{4,0} u_{1,-1}+u_{2,-1} u_{4,0}+2 u_{3,0}+u_{6,0}\right) \lambda^{2}\nonumber\\
						&+(-u_{1,-1} u_{3,0} u_{4,0}+u_{2,-1} u_{3,0} u_{4,0}\nonumber
						\\&-u_{8,-1} u_{4,0}^{2}-u_{1,0} u_{4,0}+u_{2,0} u_{4,0}+u_{3,0}^{2}+2 u_{3,0} u_{6,0}) \lambda \nonumber
						\\&-u_{1,0} u_{3,0} u_{4,0}+u_{2,0} u_{3,0} u_{4,0}+u_{3,0}^{2} u_{6,0}-u_{4,0}^{2} u_{8,0}\,,\\
						\mu_k=&-\frac{\lambda_k^2+(u_{2,-1}u_{4,0}+u_{3,0}+u_{6,0})\lambda_k+u_{2,0}u_{4,0}+u_{3,0}u_{6,0}}{u_{4,0}}\,,\quad k=1,2,3\,.
					\end{align}
					
					In this case, we can choose the holomorphic differentials as
					\begin{align}
						&\alpha_i=\frac{\lambda^{i-1}d\lambda}{S_{\mu}(\lambda,\mu)}\,,\quad i=1,\dots,2m\,,\\
						&\alpha_i=\frac{\lambda^{i-2m-1}\mu d\lambda}{S_{\mu}(\lambda,\mu)}\,, \quad i=2m+1,\dots,3m\,.
					\end{align}
					Let 
					\begin{equation}
						A=\begin{pmatrix}
							\frac{1}{S_{\mu}(\lambda_1,\mu_1)} & \dots & \frac{1}{S_{\mu}(\lambda_g,\mu_g)} \\
							\frac{\lambda_1}{S_{\mu}(\lambda_1,\mu_1)} & \dots & \frac{\lambda_g}{S_{\mu}(\lambda_g,\mu_g)} \\
							\vdots &  & \vdots \\
							\frac{\lambda_1^{2m-1}}{S_{\mu}(\lambda_1,\mu_1)} & \dots & \frac{\lambda_g^{2m-1}}{S_{\mu}(\lambda_g,\mu_g)} \\
							\frac{\mu_1}{S_{\mu}(\lambda_1,\mu_1)} & \dots & \frac{\mu_g}{S_{\mu}(\lambda_g,\mu_g)} \\
							\vdots &  & \vdots \\
							\frac{\mu_1\lambda_1^{m-1}}{S_{\mu}(\lambda_1,\mu_1)} & \dots & \frac{\mu_g\lambda_g^{m-1}}{S_{\mu}(\lambda_g,\mu_g)}
						\end{pmatrix}\,.
					\end{equation}
					The  evolution of $\mathcal{D}$ can be described as 
					\begin{equation}
						\begin{pmatrix}
							\frac{d\lambda_1}{dt_k} \\
							\vdots \\
							\frac{d\lambda_g}{dt_k}
						\end{pmatrix}=A^{-1}\bm U^{(k)}\,.
					\end{equation}
					Here $\bm U^{(k)}$ are given by equation \eqref{UkVec}. When $m=1$ we write down the first flow.
					\begin{equation}\label{A2Dubrovinm1}
						\frac{d \lambda_i}{dt_1}=\frac{S_{\mu}(\lambda_i,\mu_i)}{\prod_{j\neq i}(\lambda_i-\lambda_j)}\left(\sum_{k=1}^{3}\frac{\mu_k}{\prod_{j\neq k}(\lambda_{j}-\lambda_{k})}\right)^{-1}\,,\quad i=1,\dots,3\,.
					\end{equation}
					We list first two $\tau$-structures restricted on $W(\lambda)$ when $m=1$.
					\begin{align}
						F_{1,0;1,0}&=
						-\frac{1}{3} u_{8,-1}
						-\frac{1}{3} u_{5,-1}
						-\frac{1}{3} u_{1,-1}^{2}
						-\frac{1}{3} u_{2,-1}^{2}
						-\frac{1}{3} u_{1,-1} u_{2,-1}
						+\frac{2}{3} u_{4,0}\\
						&=-\frac{1}{3}h_{1,-1}-\left(\sum_{k=1}^{3}\frac{\mu_k}{\prod_{j\neq k}(\lambda_{j}-\lambda_{k})}\right)^{-1}\,,\\
						F_{1,0;2,0}&=
						\frac{1}{3} u_{3,0}
						+\frac{1}{3} u_{6,0}
						+\frac{1}{3} u_{2,-1} u_{4,0}
						-\frac{2}{3} u_{1,-1}^{2} u_{2,-1}
						-\frac{2}{3} u_{1,-1} u_{2,-1}^{2}\nonumber\\
						&+\frac{2}{3} u_{8,-1} u_{1,-1}
						+\frac{2}{3} u_{8,-1} u_{2,-1}
						-\frac{2}{3} u_{5,-1} u_{1,-1}
						-\frac{2}{3} u_{7,-1}\\
						&=-\frac{2}{3}h_{2,-1}-\left(\sum_{k=1}^{3}\frac{\mu_k\sum_{j\neq k}\lambda_j}{\prod_{j\neq k}(\lambda_j-\lambda_k)}\right)\left(\sum_{k=1}^{3}\frac{\mu_k}{\prod_{j\neq k}(\lambda_{j}-\lambda_{k})}\right)^{-1}\,.
					\end{align}
					Here $h_{1,-1},h_{2,-1}$ are hamiltonians and the explicit formulas are
					\begin{align}
						h_{1,-1}&=
						u_{1,-1}^{2}
						+u_{1,-1} u_{2,-1}
						+u_{2,-1}^{2}
						+u_{4,0}
						+u_{5,-1}
						+u_{8,-1}\,,\\
						h_{2,-1}&=
						u_{1,-1}^{2} u_{2,-1}
						+u_{1,-1} u_{2,-1}^{2}
						+u_{5,-1} u_{1,-1}
						-u_{8,-1} u_{1,-1}\nonumber
						\\&+u_{2,-1} u_{4,0}
						-u_{8,-1} u_{2,-1}
						+u_{3,0}
						+u_{6,0}
						+u_{7,-1}\,.
					\end{align}
					Let $r_1=F_{1,0;1,0},r_2=F_{1,0;2,0}$. They satisfy the  Boussinesq equation written in normal coordinates. 
					\begin{align}
						&\frac{\partial r_1}{\partial t_2}=\frac{\partial r_2}{\partial t_1}\,,\label{normalA2r1}\\
						&\frac{\partial r_2}{\partial t_2}=-\frac{1}{3}	\frac{\partial^3 r_{1}}{\partial(t_1)^3}-4r_1\frac{\partial r_1}{\partial t_1}\,.\label{normalA2r2}
					\end{align}
					Let
					\begin{align}
						v_1&=3r_1\,,\\
						v_2&=\frac{3}{2}\frac{\partial r_1}{\partial t_1}+\frac{3}{2}r_2\,.
					\end{align}
					We get
					\begin{align}
						\frac{\partial v_1}{\partial t_2}=&-\frac{\partial^2 v_1}{\partial (t_1)^2}+2\frac{\partial v_2}{\partial t_1}\,,\\
						\frac{\partial v_2}{\partial t_2}=&\frac{\partial^2 v_2}{\partial (t_1)^2}-\frac{2}{3}\frac{\partial^3 v_1}{\partial (t_1)^3}-\frac{2}{3}v_1\frac{\partial v_1}{\partial t_1}	\,.
					\end{align}
					These are the Gelfand--Dickey hierarchy \eqref{GelfandDickey} for $t_2$, where $L=\partial^3+v_1\partial+v_2$.
					 Equations \eqref{A2Dubrovinm1} can be viewed as the Dubrovin equations of the Boussinesq hierarchy. 
					
					In \cite{Dickson1999,Dickson1999b}, Dickson--Gesztesy--Unterkofler obtained two series of algebro-geometric solutions to the Boussinesq hierarchy, which are related to algebraic curves of genus $g=3m$ and $g=3m+1$, $m\geq 1$. In this paper only the solutions related to genus $g=3m$ curves are involved.

				In particular, when $m=1$, and consider  
				\begin{equation}
						W(\lambda)=\begin{pmatrix}
							-\frac{62}{3\lambda} & 1 & \frac{4}{\lambda} \\
							-4 & -\frac{8}{3\lambda} & 1 \\
							\lambda-\frac{121}{\lambda} & 7 & \frac{70}{3\lambda}
						\end{pmatrix}\,.
					\end{equation}
					The spectral curve of $\lambda W(\lambda)$ is 
					\begin{equation}
						C: \mu^3-(7\lambda^2+\frac{16}{3})\mu-\lambda^4-\frac{47\lambda^2}{3}+\frac{128}{27}=0\,.
					\end{equation}
					Unlike the general cases, there are three nodal points on the affine part of $C$. Therefore $C$ is a rational curve, and can be parametrized by $w\in \mathbb{P}^1$
					\begin{equation}
						\left\{\begin{aligned}
							\lambda&=w^3-4w\\
							\mu&=w^4-3w^2-\frac{8}{3}\,.
						\end{aligned}\right.
					\end{equation}
					This map is one-to-one except for the nodal points. For each nodal point on $C$, there are two values of $w$ corresponding to it, and we denote these two values by $b_j,c_j$, $j=1,2,3$.
					In this example:
					\begin{equation}
						b_1=-2,\, b_2=-\frac{1}{2}-\frac{\sqrt{5}}{2},\, b_3=-\frac{1}{2}+\frac{\sqrt{5}}{2}\,,c_1=2,\, c_2=\frac{1}{2}-\frac{\sqrt{5}}{2},\, c_3=\frac{1}{2}+\frac{\sqrt{5}}{2}\,.
					\end{equation}
					And the $i$-th nodal point on $C$  is the point with 
					\begin{equation}
						w=b_i,c_i\,,\quad i=1,2,3\,.
					\end{equation}
					
					From the viewpoint of Mumford \cite{Mumford83}, we can treat this curve as a limit of a genus $3$ curve. By choosing suitable cycles $a_i,b_i$, $i=1,2,3$. The limit of entries of the periodic matrix are
					\begin{equation}
						B_{ij}=\log\frac{(b_j-b_i)(c_j-c_i)}{(b_j-c_i)(c_j-b_i)}\,,\quad i\neq j \,.
					\end{equation}
					When $i=j$, we have
					\begin{equation}
						{\rm Re}(B_{jj})\rightarrow -\infty\,.
					\end{equation}
					For any $\bm u\in \mathbb{C}^g$, denote $\tilde{\bm u}=(u_1-\frac{B_{11}}{2},\dots,u_g-\frac{B_{gg}}{2})$. The limit of the $\theta$-function, are the $\theta$-function of the generalized Jacobian for singular curve $C$, introduced by Mumford \cite{Mumford83} (see also \cite{Kodama24}).
					\begin{equation}
						\theta(\tilde{\bm u})\rightarrow\theta_M(\bm u)
					\end{equation}
					Here 
					\begin{equation}
						\theta_M({u}_1,{u}_2,{u}_3)=\sum_{m\in\{0,1\}^3}\exp\left(\sum_{i<j}B_{ij}m_im_j+\sum_{j=1}^3{u}_jm_j\right)\,.
					\end{equation}
					
					We can still consider the divisor of poles from this $W(\lambda)$:
					\begin{equation}
						\mathcal{D}=Q_1+Q_2+Q_3\,.
					\end{equation}
					In this example, we have
					\begin{equation}
						w(Q_1)=0,\,w(Q_2)=\sqrt{2},\,w(Q_3)=-\sqrt{2}\,.
					\end{equation}
					On this singular curve, the normalized holomorphic differential forms become
					\begin{equation}
						\omega_j=\frac{dw}{w-b_j}-\frac{dw}{w-c_j}\,,\quad i=1,\dots,3\,.
					\end{equation}
					 The point $\bm u$ in $J(C)$ is
					\begin{equation}
						\bm u=\mathcal{D}-\infty_C-\Delta=\sum_{i=1}^{3}\int_{\infty_C}^{Q_i}\bm\omega+\bm K_{\infty}\,.
					\end{equation}
					Here $\bm \omega=(\omega_{1},\dots,\omega_{3})$ and $\bm K_{\infty}=2\infty_C-\Delta$ is the Riemann constant, see \cite{Fay,Dubrovin1981}. The limit of $\bm K_{\infty}$ is 
					\begin{equation}
						K_j-\frac{B_{jj}}{2}\rightarrow \tilde{K}_{j}=\pi\sqrt{-1}-\sum_{k\neq j}\log\frac{c_k-b_j}{c_k-c_j}\,.
					\end{equation}
					Therefore the limit of $ u_j-\frac{B_{jj}}{2}$, $j=1,2,3$ are
					\begin{equation}
						u_j-\frac{B_{jj}}{2}\rightarrow \tilde{u}_j=\sum_{k=1}^{3}\log\frac{w(Q_k)-b_j}{w(Q_k)-c_j}+\pi\sqrt{-1}-\sum_{k\neq j}\log\frac{c_k-b_j}{c_k-c_j}\,.
					\end{equation}
					To get $Z(\bm t)$, we still need to compute the vector $\bm V^{(k)}$ and the constants $q_{k,l}$.
					For $\bm V^{(k)}$, equation \eqref{expand_omega} still holds, which is
					\begin{equation}
						\bm\omega(\lambda)=\frac{1}{3}\sum_{k\geq 1}\bm V^{(k)}\lambda^{-\frac{k}{3}-1}d\lambda\,.
					\end{equation}
					From this formula we get the vector $\bm V^{(k)}$\,.
					The  fundamental normalized bi-differential on $C$ is 
					\begin{equation}
						\omega(w_1,w_2)=\frac{dw_1dw_2}{(w_1-w_2)^2}\,.
					\end{equation}
					Using \eqref{expand_bidiffer} we have following equation for computing $q_{j,k}$:
					\begin{equation}
						\frac{dw_1dw_2}{(w_1-w_2)^2}=\frac{dz_1dz_2}{(z_1-z_2)^2}+\sum_{j,k\geq 1}q_{j,k}z_1^{j-1}z_2^{k-1}dz_1dz_2\,.
					\end{equation}
					Here 
					\begin{equation}
						z=\lambda^{-\frac{1}{3}}=(w^3-4w)^{-\frac{1}{3}}\,.
					\end{equation}
					The $\tau$-function is therefore
					\begin{equation}
						Z(\bm t)=e^{\sum_{k,l\geq1}\frac{1}{2}q_{k,l}t_kt_l}\theta_M(\tilde{\bm u}-\sum_{k\geq 1}t_k\bm V^{(k)})\,.
					\end{equation}
					As we have showed, this $\tau$-function is a limit of that in \eqref{GDtau}.
					
					In this particular example, let $t_1=x, t_2=t$, and other $t_k$ equals to $0$. The $\tau$-function is 
					\begin{align}
						Z(x,t)=e^{-\frac{2x^2}{3}-\frac{16t^2}{9}}\bigl(&1+5e^{4x}+2e^{-\sqrt{5}t+x}+2e^{-\sqrt{5}t+5x}+2e^{\sqrt{5}t+x}\nonumber\\
						&+2e^{\sqrt{5}t+5x}+5e^{2x}+e^{6x}\bigr)\,.
					\end{align}

					Let $r_1=\frac{\partial^2 \log Z(x,t)}{\partial x^2}$, $r_2=\frac{\partial^2\log Z(x,t)}{\partial x\partial t}$, we can verify that they satisfy the  Boussinesq equation \eqref{normalA2r1} and \eqref{normalA2r2}.  This is a regular $3$-soliton solution. For a detailed description of regular solitons of the Gelfand--Dickey hierarchy, one can refer to \cite{Kodama25}.  One can also verify that the $N$-th order Taylor coeffients of $\log Z(x,t)$ are all rational numbers.
					
				\section*{Acknowledgments}
					I would like to thank professor Di Yang for his advice and helpful discussions, and professor Cheng Zhang for helpful suggestions. The work is  supported by NSFC No. 12371254 and CAS No. YSBR-032.

				\end{document}